\documentclass[reqno]{amsart}
\usepackage{tkz-euclide}
\usepackage{hyperref}
\usepackage{amsmath}  
\usepackage{amsfonts}
\usepackage{amssymb} 
\usepackage{mathrsfs} 
\usepackage{graphicx}  
\usepackage{bm}
\usepackage{tikz} 
\usetikzlibrary{svg.path} 
\usepackage{amssymb}
\usepackage{pgfplots}

\usetikzlibrary{calc} 
\usepackage{xcolor}  
\usetikzlibrary{arrows.meta} 
\usepackage{amsthm} 
\usepackage{float}
\usepackage{enumitem}
\usepackage[english]{babel}
\usepackage[font=footnotesize, labelfont=bf]{ caption}
\usepackage{ esint }
\usepackage{stackengine,scalerel,graphicx}
\stackMath

\definecolor{trueblue}{rgb}{0.0, 0.45, 0.81}
\definecolor{truegreen}{rgb}{0.13, 0.55, 0.13}

\newcommand{\BBB}{\color{black}}

\newcommand{\EEE}{\color{black}}

\newcommand{\eps}{\varepsilon}

\theoremstyle{plain}
\begingroup
\newtheorem{theorem}{Theorem}[section]
\newtheorem{lemma}[theorem]{Lemma}

\endgroup

\newenvironment{step}[1]{\underline{Step #1}.}{}

\theoremstyle{definition}
\begingroup
\newtheorem{definition}[theorem]{Definition}
\endgroup

\numberwithin{equation}{section}
\newcommand{\N}{\mathbb{N}}
\newcommand{\Z}{\mathbb{Z}}
\newcommand{\R}{\mathbb{R}}

\newcommand{\x}{{\times}}

\usepackage[normalem]{ulem}

\begin{document}

\title[A proof of finite crystallization via strataficition]{A proof of finite crystallization via stratification}

\author[M. Friedrich]{Manuel Friedrich} 
\address[Manuel Friedrich]{Department of Mathematics, Friedrich-Alexander Universit\"at Erlangen-N\"urnberg. Cauerstr.~11,
D-91058 Erlangen, Germany, \& Mathematics M\"{u}nster,  
University of M\"{u}nster, Einsteinstr.~62, D-48149 M\"{u}nster, Germany}
\email{manuel.friedrich@fau.de}

\BBB

\author[L.~Kreutz]{Leonard Kreutz}
\address[Leonard Kreutz]{Zentrum Mathematik - M7, Technische Universit\"at M\"unchen, Garching, Germany}
\email{leonard.kreutz@.tum.de}

\EEE


\begin{abstract} We devise a new technique to prove two-dimensional crystallization results in the square lattice for finite particle systems. We apply this strategy to energy minimizers of configurational energies featuring two-body short-ranged particle interactions and three-body angular potentials favoring bond-angles of the square lattice. To each  configuration, we  associate its bond graph which is then suitably modified by identifying chains of successive atoms. This method, called \emph{stratification},  reduces  the crystallization problem to a simple minimization  that corresponds to a proof via slicing of the isoperimetric inequality in $\ell^1$. As a byproduct, we also prove a fluctuation estimate  for minimizers of the configurational energy, known as the  $n^{3/4}$-law.  

\end{abstract}

\subjclass[2010]{}

\maketitle

\section{Introduction}

At low temperature, atoms and molecules typically arrange themselves into crystalline order. Tackling this phenomenon by using  mathematical models consists in proving or disproving that ground states  of  particle systems for certain configurational energies with interatomic interactions exhibit crystalline order. This issue, referred to  as the {crystallization problem} \cite{Blanc},  has attracted  a great deal of attention in the physics and mathematics community. By now,  various mathematically rigorous crystallization results are available both for systems with a fixed, \emph{finite} number of atoms, and in  the so-called \emph{thermodynamic limit} dealing with the infinite particle limit. The reader is referred  to  \cite{Blanc, Friesecke-Theil15} for a
general overview and also to \cite{Mainini-Piovano} for a detailed account of available results.   The goal of this paper is to revisit the problem of finite crystallization in dimension two, and to present a novel and substantially different proof strategy.

We consider a model where configurations are identified  with the respective positions of  atoms  $\lbrace x_1, \ldots, x_n \rbrace$ in the plane with an associated configurational energy $\mathcal{E}( \lbrace x_1,\ldots,x_n \rbrace )$  
comprising classical interaction potentials. More specifically, $\mathcal{E} = \mathcal{E}_2 + \mathcal{E}_3$ decomposes into $\mathcal{E}_2$ and $\mathcal{E}_3$  describing two- and three-body interactions, respectively.  The two-body interaction potential $\mathcal{E}_2$ is short-ranged
and attractive-repulsive favoring atoms sitting at some specific reference distance. For $\mathcal{E}_3 \equiv 0$ and for a specific choice of $\mathcal{E}_2$, namely the so-called \emph{sticky disc potential}, crystallization in the triangular lattice has been proved by {\sc Heitmann \& Radin}   \cite{HR} (see also \cite{Radin, Wagner83} for generalizations) and recently revisited in  \cite{Lucia}, via an approach from discrete differential geometry. If instead \BBB $\mathcal{E}_3\neq 0$, under specific quantitative assumptions, \EEE optimal geometries can be identified as the square or the hexagonal lattice \cite{Mainini-Piovano, Mainini},  depending on whether $\mathcal{E}_3$ favors triples of particles forming angles which are multiples of $\frac{\pi}{2}$ or $\frac{2\pi}{3}$, respectively.   Besides crystallization, fine characterizations of ground-state geometries are available by proving the emergence of hexagonal or  square macroscopic Wulff shapes for growing particle numbers \cite{Yuen, Davoli15, Davoli16, FriedrichKreutzSchmidt}. We also refer to related rigorous crystallization results for particle systems involving  different types of atoms \cite{Betermin, B43, B195, FriedrichKreutzHexagonal, FriedrichKreutzSquare, FriedrichStefanelli}  and to     \cite{Gardner,Hamrick,Radi,Ventevogel} for a nonexhaustive list of results in dimension one.

Although the exact realization of the proof of each result is different depending on the used potentials and the underlying optimal geometry, all proofs follow the very same strategy, originally devised in  \cite{Harborth, HR}. \BBB First of all, due to range of the two-body interaction, one can naturally associate a planar graph to the configuration where vertices and edges correspond to particles and bonds, respectively. \EEE The graph is then separated into the boundary and bulk atoms. The  \emph{boundary energy} (roughly, the number of bonds at the boundary of the configuration) is carefully estimated by geometric arguments involving the angles between atoms and relying on the sum of interior angles in planar polygons. Moreover, by means of  Euler's formula for planar graphs a connection between the number of bonds and atoms in the configuration is derived. Then, the essential idea of the proof lies in an induction argument over the number of particles: one removes a \emph{bond graph layer}, i.e., the boundary atoms of the configuration, and by induction hypothesis one uses information of the remaining configuration consisting of less atoms. The approach in  \cite{Lucia} is different in the sense that it endows the bond graph with a suitable notion of discrete combinatorial curvature and uses a discrete version of the Gauss-Bonnet theorem from differential geometry. However, it still vitally hinges on specific geometric arguments and the   induction method over bond graph layers. 

It appears to be challenging to generalize this strategy  to problems beyond the setting described above. On the one hand, it is hardly conceivable to extend the delicate estimates on the boundary  energy to particle systems in three dimensions where surfaces have a much richer structure. On the other hand, the induction method over bond graph layers is often not flexible enough to handle more general situations such as particles systems with two types of atoms with prescribed ratio since this ratio might not be preserved by removing a bond graph layer.  
 
In this work, we propose a new strategy to tackle finite crystallization problems which does not use the induction method over bond graph layers and comes along without arguments from the theory of planar graphs and discrete  differential geometry  such as Euler's formula or Gauss-Bonnet. It relies on an idea that we call \emph{stratification}. In this paper, we present our technique  for the model by {\sc Mainini, Piovano, \& Stefanelli} \cite{Mainini-Piovano} and reprove finite crystallization in the square lattice, see Theorem \ref{thm:crystallization}. We are confident, however, that the strategy carries over to other lattices as well, such as the triangular \cite{Lucia, HR} and the hexagonal \cite{Mainini} lattice. 

As observed in \cite{Mainini-Piovano}, ground states correspond to configurations minimizing a specific edge perimeter of the configuration, essentially counting the number of missing bonds of atoms having less than four bonds. For ground-state competitors,  the bond graph can be locally interpreted as a deformed version of $\Z^2$, apart from possible defects in the lattice, see Definition~\ref{def:epsregular}. Therefore, we can identify chains of atoms in the bond graph where the angle between three successive atoms is near to $\pi$, called \emph{strata}. Strata can be \emph{open}, where the first and the last atom  of the stratum lie at the boundary of the configuration, or  they can be \emph{closed} forming a closed cycle. In contrast to open strata, closed strata do not contribute to the edge perimeter. Therefore, for a correct estimate, we aim at excluding the existence of closed strata. To this end, we observe that, due to the cycle structure of closed strata, there need to exist angles deviating from $\pi$ and thus contributing to the  three-body energy $\mathcal{E}_3$. Given specific quantitative assumptions on the potentials similar to the ones in  \cite{Mainini-Piovano}, the contribution of $\mathcal{E}_3$ is large enough to allow us to erase a bond from the stratum to turn it into an open stratum. This procedure is made precise in Lemma \ref{lemma:construction} and referred to as \emph{stratification}. Once all strata are open, the graph satisfies specific properties (see Lemmas \ref{lemma:opengraph} and \ref{lemma:excess}) which  reduce our crystallization problem  to a simple argument related to an edge isoperimetric inequality on the square lattice. (Compare to \cite{Mainini-Piovano} for a problem on $\Z^2$, and see also \cite{Bollobas, Harary} for  some related classical issues in Discrete Mathematics.)

In contrast to uniqueness  of Wulff shapes for continuum crystalline isoperimetric problems, minimizers for a finite number of particles $n$ are in general not unique. For different lattices in 2D, it has been shown that there are arbitrarily large $n$  with ground-state configurations
deviating from the hexagonal or  square macroscopic Wulff shape by a number of $n^{3/4}$-particles  \cite{Davoli15, Davoli16, Mainini-Piovano, Schmidt-disk}. Later, this analysis as been extended to the cubic lattice in higher dimensions  \cite{Mainini-Piovano-schmidt, edo-bernd}.  The proof of such maximal asymptotic deviation, also known as \emph{maximal  fluctuation estimate}, relies on careful rearrangement techniques for atoms at the boundary and edge-isoperimetric inequalities. In our setting of the square lattice, we can immediately reobtain this so-called $n^{3/4}$-law as a mere byproduct of our crystallization proof,  see Theorem \ref{th: 34law}.  Our argument is similar to \cite{Mainini-Piovano} with the interesting difference however that our strategy  can be applied even if configurations are \textit{not} subset of $\Z^2$. We also mention the complementary approach  \cite{cicalese}, yet restricted to subsets of periodic lattices, where maximal  fluctuation estimates are derived via a quantitative version of the edge isoperimetric inequality, based on the
quantitative version of the anisotropic isoperimetric inequality proved in \cite{Figalli-Maggi-Pratelli}. 

One goal of our work is to revisit finite crystallization results and to suggest a substantially different proof strategy which does not use the induction method over bond graph layers and comes along without arguments from the theory of planar graphs and discrete  differential geometry. Besides providing, to our view, a simpler and more direct proof of known results, our main motivation is that our techniques seem promising to tackle more challenging crystallization problems. For example, we expect that our approach can contribute to understand  finite  crystallization   in three dimensions  or crystallization for double-bubble problems  \cite{Duncan0, Duncan, gorny} (configuration with two types of atoms).


Let us highlight that our proof strategy is tailor-made for the problem of finite crystallization. Concerning   the thermodynamic limit,  i.e., as the number of particles tends to infinity, other techniques are used  and allow to prove results under less restrictive assumptions on the potentials. We refer to \cite{Betermin0,ELi, Smereka15, Theil}   for results in the plane and to some few available rigorous results \cite{Flateley1,Flateley2} in three dimensions.

The article is organized as follows. In Section~\ref{sec:setting} we introduce our setting and state the main results.   Section~\ref{sec:mainproof} is devoted to the concept of stratification and in Section \ref{sec: main} we prove our main results. We close the introduction with basic notation. The Euclidian distance between a point $x \in \R^d$ and a set $A \subset \R^d$ is denoted by ${\rm dist}(A,x)$. By $\# A$ we denote the cardinality of a set $A$. By $B_r(x)$ we indicate the open ball with center $x \in \R^d$ and radius $r >0$, and simply write $B_r$ if $x = 0$. We define the ceil function by   $\lceil t \rceil := \min \lbrace z \in \mathbb{Z} \colon z \ge t\rbrace$ for $t \in \R$.

\section{Setting and main results}\label{sec:setting}

We consider particle systems in two dimensions, and model their interaction by classical potentials in the frame of Molecular Mechanics \cite{Molecular, Lewars}. Indicating  the configuration of particles by $C_n =\{x_1,\ldots,x_n\} \subset \mathbb{R}^2$, we define its energy by
\begin{align}\label{eq: main energy}
\mathcal{F}(C_n) = \frac{1}{2}\sum_{i\neq j} v_2(|x_i-x_j|) + \frac{1}{2} \sum_{i,j,k} v_3(\theta_{i,j,k})\,.
\end{align}
Here, $\theta_{i,j,k}$ denotes the angle formed by the vectors $x_j-x_i$ and $x_k-x_i$ (counted clockwisely), and the second sum runs over triples $(i,j,k)$ \BBB with  $|x_i-x_j| \le r_0$ and $|x_i-x_k| \le r_0$, where $r_0$ is given in ($\rm ii_2$) below. \EEE    The factor $\frac{1}{2}$ accounts for double counting of bonds and angles. In the following, for simplicity we denote the angle formed by the vectors $x-y$ and $z-y$ by  $\theta_{x,y,z}$. We fix  $0<\varepsilon<\varepsilon_0$ for $\varepsilon_0 < \frac{\pi}{6}$ specified in Lemma~\ref{lemma:elementaryprop}. The two-body potential $v_2 \colon [0,+\infty) \to {\mathbb{R}}  \cup \lbrace + \infty \rbrace $  satisfies 
\begin{itemize}
\item[($\rm i_2$)] $\min_{r \geq 0} v_2(r) = v_2(1)=-1$ and $v_2(r) >-1$ if $r \neq 1$;
\item[($\rm ii_2$)] There exists $1<r_0<\sqrt{2}$ such that $v_2(r)=0$ for all $r \geq r_0$;
\item[($\rm iii_2$)] For all $r \in [0,1-\varepsilon] $  it holds that   $v_2(r) > \varepsilon^{-1}$.
\end{itemize}
The three-body potential $v_3 \colon [0,2\pi] \to \mathbb{R}$ satisfies
\begin{itemize}
\item[($\rm i_3$)] $v_3(\theta)=v_3(2\pi-\theta)$ for all $\theta \in [0,2\pi]$;
\item[($\rm ii_3$)] $v_3(k\pi/2)=0$ for $k=1,2,3$ and $v_3(\theta) >0$ if $\theta \notin \{\pi/2,\pi,3\pi/2\}$;
\item[($\rm iii_3$)] $v_3(\theta)  \geq  \EEE 4(\pi/6-\varepsilon)^{-1} \, |\theta -\pi |$ for all $\theta \in [\pi- \eps,\pi + \eps]$     with equality only if $\theta = \pi$;  
\item[($\rm iv_3$)]  $\theta \notin [\pi/2-\varepsilon,\pi/2+\varepsilon] \cup [\pi-\varepsilon,\pi+\varepsilon] \cup [3\pi/2-\varepsilon,3\pi/2+\varepsilon] \implies v_3(\theta) > \frac{4}{(1-\varepsilon)^2}(\sqrt{2} +\frac{1}{2})^2$\,.
\end{itemize}
We briefly comment on the assumptions. Condition ($\rm i_2$) on a unique minimum (here normalized to $1$)   is natural, e.g., it is valid for Lennard-Jones-type potentials.  Assumption ($\rm ii_2$) states that $v_2$ has compact support. In particular, it ensures that for configurations $ C_n  \subset \mathbb{Z}^2$ only atoms at distance $1$ interact.  These atoms are  usually  referred to as \emph{nearest neighbors} in the literature\EEE. Eventually, ($\rm iii_2$) prevents clustering of points. In fact,  along with ($\rm ii_2$) it shows Definition \ref{def:epsregular}(i) in the proof of Lemma \ref{lemma:elementaryprop} below.  Condition ($\rm i_3$) ensures that the potential $v_3$ does not depend on how (clockwise or counter-clockwise) bond angles are measured, and ($\rm ii_3$) guarantees  that for $ C_n \subset \mathbb{Z}^2$ there is no contribution stemming from the three-body interaction. Slope conditions similar to   ($\rm iii_3$) have been used in \cite{FriedrichKreutzHexagonal,FriedrichKreutzSquare,Mainini-Piovano, Mainini} in order to obtain crystallization on the square or hexagonal lattice. Let us mention that in the  other works the condition is needed at  \emph{all} minimum points of $v_3$, whereas here only at $\pi$. As a consequence,  the potential is necessarily non-smooth  at $\pi$.  We also point out that in this work the focus lies on a new proof strategy and all appearing  specific numerical constants   are chosen for computational simplitcity rather than optimality.  The two potentials are illustrated in Figure~\ref{fig:potentials}.

\begin{figure}[htp]
\begin{tikzpicture}
\tikzset{>={Latex[width=1mm,length=1mm]}};

\draw[dash pattern=on 1.5pt off 1pt,ultra thin](1,-.04)--(1,-.475);

\draw[->](0,-1)--++(0,4) node[anchor= east] {$v_{2}(r)$};
\draw[->](-1,0)--++(4.5,0) node[anchor =north] {$r$};

\draw(2.6,-.04)--++(0,.08) node[anchor=south]{$r_0$}; 

\draw(1,-.04)--++(0,.08);

\draw(1,0) node[anchor=south]{$1$}; 

\draw[dash pattern=on 1.5pt off 1pt,ultra thin](0,-.475)--(1,-.475);

\draw(0,-.475) node[anchor=east]{$-1$};

\draw[thick,domain=.725:2.5, smooth, variable=\x] plot ({\x}, {.025+.5/(\x*\x*\x*\x*\x*\x*\x*\x)-1/(\x*\x*\x*\x)});
\draw[thick](2.49,0)--++(0:.8);

\begin{scope}[shift={(7,0)}]

\draw[dash pattern=on 1.5pt off 1pt,ultra thin](2,0)--++(110:.5);

\draw[dash pattern=on 1.5pt off 1pt,ultra thin](2,0)--++(110:-.25);
\draw[->](0,-1)--++(0,4) node[anchor= east] {$v_{3}(\theta)$};
\draw[->](-1,0)--++(6,0) node[anchor =north] {$\theta$};

\draw(1,-.04)--++(0,.08); 
\draw(3,-.04)--++(0,.08); 

\draw(1,0) node[anchor=north]{$\pi/2$};

\draw(3,0) node[anchor=north]{$3\pi/2$};

\draw(4,-.04)--++(0,.08); 
\draw(2,-.04)--++(0,.08); 
\draw(4,0) node[anchor=north]{$2\pi$};
\draw(2,-.06) node[anchor=north]{$\pi$};

\draw[thick,domain=2.93333:3.04196, smooth, variable=\x] plot ({\x}, {70*(\x-3)*(\x-3)});

\draw[thick,domain=3.04196:4, smooth, variable=\x] plot ({\x}, {1.5*(-.21 - (\x-2.9)*(\x-5.1))});

\draw[thick,domain=0:0.958042, smooth, variable=\x] plot ({\x}, {1.5*(-.21 - (\x-1.1)*(\x+1.1))});

\draw[thick,domain=2:2.93333, smooth, variable=\x] plot ({\x}, {5*(-.11 - (\x-1.9)*(\x-3.1))});

\draw[thick,domain=0.958042:1.06667, smooth, variable=\x] plot ({\x}, {70*(\x-1)*(\x-1)});

\draw[thick,domain=1.06667:2, smooth, variable=\x] plot ({\x}, {5*(-.11 - (\x-.9)*(\x-2.1))});
\end{scope}
\end{tikzpicture}
\caption{The potentials $v_2$ and $v_3$.}
\label{fig:potentials}
\end{figure}
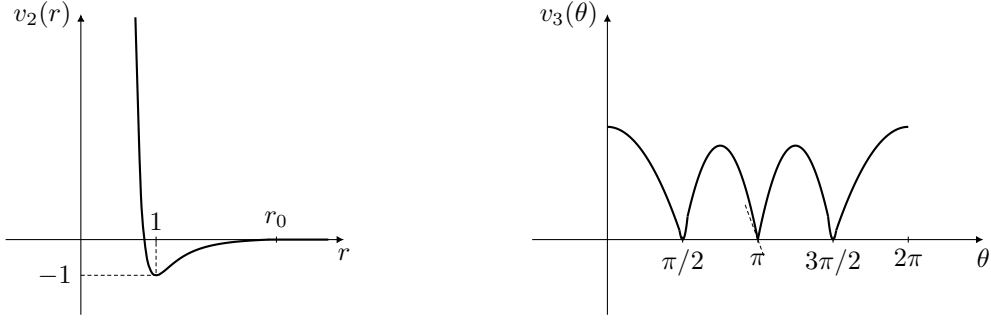

We now  state the main theorems of the paper. \BBB We emphasize that the main theorems have been shown previously in the literature, see \cite{FriedrichKreutzSquare,Mainini-Piovano}.  As outlined in the introduction, the main novelty lies in the proof technique. \EEE

\begin{theorem}[Crystallization]\label{thm:crystallization}
 For each $C_n \in (\R^2)^n$, it holds that 
 \begin{align}\label{eq: main claim0}
 \mathcal{F}(C_n) \geq  - 2n + \EEE \lceil 2\sqrt{n}\rceil\,.
\end{align}
\BBB Equality in \eqref{eq: main claim0} implies that \EEE $C_n \subset \mathbb{Z}^2$  (up to a rigid motion).
\end{theorem}

 Some configurations of minimal energy are depicted in Figure~\ref{fig:groundstates}.

\begin{theorem}[$n^{3/4}$-law]\label{th: 34law}
There exists $c>0$ such that for all $n \in \N$ it holds that each ground state $C_n$, up to  a rigid motion, satisfies    
$$\# \big( C_n \triangle  S_n      \big) \le cn^{3/4}, $$
where $S_n :=  [ 1,\lceil \sqrt{n} \rceil ]^2 \EEE \cap \Z^2$.
\end{theorem}
Let us note that the scaling is sharp: the construction in  \cite[Section 3.2]{FriedrichKreutzSquare} shows that there exists a sequence $(n_k)_{k \in \N}$ with $n_k \to +\infty$ and corresponding ground states $C_{n_k}$  such that, up to applying any rigid motion to $ C_{n_k}$, it holds that \BBB
\begin{align*}
\# \big( C_{n_k} \triangle  S_{n_k}   \big) \ge \overline{c}n_k^{3/4} 
\end{align*}
for some $0<\overline{c}\leq c$, where $c$ is the constant given in Theorem \ref{th: 34law}.  Our proof allows to give an explicit estimate on the constant $c$, but we do not know if $\overline{c}=c$.  \EEE

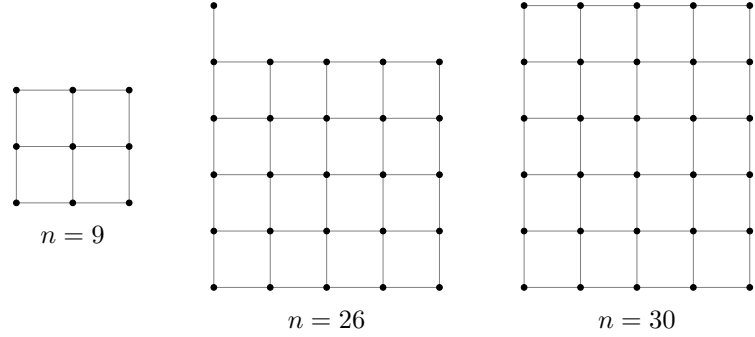
\begin{figure}[htp]
\begin{tikzpicture}[scale=0.75]

\begin{scope}[shift={(-5,0)}]

\draw(1,-.25) node[anchor=north] {$n=9$};

\draw[ultra thin,gray] (0,0) grid (2,2);

\foreach \j in {0,1,2}{
\foreach \i in {0,1,2}{
\draw[fill=black](\i,\j) circle(.05);

}

}

\end{scope}

\begin{scope}[shift={(-1.5,-1.5)}]

\draw(2,-.25) node[anchor=north] {$n=26$};

\draw[ultra thin,gray] (0,0) grid (4,4);

\draw[ultra thin,gray] (0,4) --++(0,1);

\draw[fill=black](0,5) circle(.05);

\foreach \j in {0,...,4}{
\foreach \i in {0,...,4}{
\draw[fill=black](\i,\j) circle(.05);

}

}

\end{scope}

\begin{scope}[shift={(4,-1.5)}]

\draw(2,-.25) node[anchor=north] {$n=30$};

\draw[ultra thin,gray] (0,0) grid (4,5);

\foreach \j in {0,...,5}{
\foreach \i in {0,...,4}{
\draw[fill=black](\i,\j) circle(.05);

}

}

\end{scope}

\end{tikzpicture}
\caption{Configurations of minimal energy for different cardinality.}
\label{fig:groundstates}
\end{figure}

\section{Stratification}\label{sec:mainproof}

After a short preliminary on graph theory, this section is devoted to the main technique of this paper: modification of bond graphs, called stratification. 

\subsection{Bond graph} \label{subsec:graph theory}  We denote by $G=(V,E)$ a graph, where $V \subset \R^2$ indicates the set of vertices and $E \subset \{\{x,y\} \colon x,y \in V \text{ and } x\neq y \}$ is  the set of edges. For $x \in V$, we denote the neighborhood with respect to $G$ by
\begin{align*}
\mathcal{N}(x,E) :=\{y \in V \colon \{x,y\} \in E\}\,.
\end{align*}
Given $G=(V,E)$ we define
\begin{align*}
F(G) = F_\mathrm{bond}(G) + F_{\mathrm{ex}}(G)\,,
\end{align*}
where
\begin{align*}
F_{\mathrm{bond}}(G) = \sum_{x \in V} (4-\#\mathcal{N}(x,E))\,
\end{align*}
 is the \textit{bond energy} and  
\begin{align*}
 F_{\mathrm{ex}}(G)  =  \sum_{\{x,y\} \in E} (v_2(|x-y|)+1) +  \sum_{ \{x,y\},\{y,z\} \in E} v_3(\theta_{x,y,z})
\end{align*}
the \textit{excess energy}. For $V'\subset V$, we also define the localized elastic energy by
\begin{align}\label{def:Felasticlocal}
F_{\mathrm{ex}}(V') = F_{\mathrm{ex}}(G[V'])\,,
\end{align}
where $G[V']$ is the (vertex) induced subgraph of $V'$ in $G$, that is $G[V']= (V',E')$ with $E'= \{\{x,y\} \in E\colon x,y \in V'\}$.

We will identify each $C_n \subset \mathbb{R}^2$ with its \emph{natural bond graph} $G_{\mathrm{nat}}=(V,E_{\mathrm{nat}})$, where $V=C_n$ and the \emph{natural edges} are given by
\begin{align}\label{eq: relation-new}
  E_\mathrm{nat}   = \{\{x,y\} \colon x, y \in C_n, |x-y|\leq r_0\}\,,
\end{align}
 for $r_0>0$ as given in {\rm ($\rm{ii}_2$)}.
 This definition is motivated by the relation to  \eqref{eq: main energy}, namely 
\begin{align}\label{eq: relation}
2\mathcal{F}(C_n) =  -4n  + F(G_{\mathrm{nat}})\,.
\end{align} 
In Subsection \ref{subsec:stratbondgraph} below,  we will successively modify $E_{\mathrm{nat}}$ to a smaller set of edges $E\subset E_\mathrm{nat}$. \EEE

\begin{definition}\label{def:epsregular} We say that $G=(V,E)$ is \emph{$\varepsilon$-regular} if:
\begin{itemize}
\item[(i)] If $\{x,y\} \in E$, then
\begin{align*}
 |x-y| \ge 1-\varepsilon\,;
\end{align*}
\item[(ii)] If $\theta$ is a bond angle, then
\begin{align*}
\theta \in [\pi/2-\varepsilon,\pi/2+\varepsilon] \cup [\pi-\varepsilon,\pi+\varepsilon] \cup [3\pi/2-\varepsilon,3\pi/2+\varepsilon]\,.
\end{align*}
\end{itemize}
\end{definition}
Note that, if $G_{\mathrm{nat}}=(V,E_{\mathrm{nat}})$ is $\varepsilon$-regular, then it is easy to see that $G=(V,E)$ is $\varepsilon$-regular for all $E\subset E_{\mathrm{nat}}$.

\begin{lemma}\label{lemma:elementaryprop} There exists $\varepsilon_0 >0$ such that the following holds true:  if $v_2$, $v_3$ satisfy {\rm ($\rm{i}_2$)}--{\rm ($\rm{iii}_2$)} and {\rm ($\rm{i}_3$)}--{\rm ($\rm{iv}_3$)} for some $0<\varepsilon<\varepsilon_0$ and  if $C_n$ is a minimizer of \eqref{eq: main energy}, then its natural bond graph $G_{\mathrm{nat}}=(V,E_{\mathrm{nat}})$ is $\varepsilon$-regular. Moreover, it holds that  $\#\mathcal{N}(x,E_{\mathrm{nat}}) \leq 4$ for all $x \in V$.
\end{lemma}

 Analogous properties have been derived in  \cite[Propostion 2.1]{Mainini-Piovano} and \cite[Lemma 2.2]{Theil}. However, as our assumptions on the potentials are slightly different,  we include a sketch of the proof for the reader's convenience in Appendix \ref{appendix}.  For the remainder of this paper, we assume that $\varepsilon_0>0$ is chosen small enough such that Lemma \ref{lemma:elementaryprop} holds true and that $v_2$, $v_3$ satisfy {\rm ($\rm{i}_2$)}--{\rm ($\rm{iii}_2$)} and {\rm ($\rm{i}_3$)}--{\rm ($\rm{iv}_3$)} for some $0<\varepsilon<\varepsilon_0$. \BBB Moreover, we suppose that $\varepsilon_0 < 1- \frac{r_0}{\sqrt{2}}$, where $r_0$ is given in {($\rm{ii}_2$)}. This ensures that the bond graph is planar. Indeed, given a quadrilateral
with all sides larger or equal than $1-\eps_0$, one diagonal has at least length $\sqrt{2}(1-\eps_0)$.\EEE

\subsection{Stratified bond graph}  \label{subsec:stratbondgraph}

Given $G=(V,E)$, we say that $\gamma=(x_1,\ldots,x_N)$ with $x_i \in V$ for all $i=1,\ldots,N$ is  a \emph{straight path} if  $N \ge 2$ and the following holds:
\begin{itemize}
\item[(i)] $\{x_i,x_{i+1}\} \in E$ for all $i=1,\ldots,N-1\,$;
\item[(ii)] $\theta_i \in [\pi-\varepsilon,\pi+\varepsilon]$ for all $i \in 2,\ldots,N-1$, where   $\theta_i =\theta_{x_{i+1},x_i,x_{i-1}}$;   
\item[(iii)] $\{x_i,x_{i+1}\} \neq \{x_{j},x_{j+1}\}$ for all $i,j =1,\ldots,N-1$, $ j \neq i$.
\end{itemize}
 (If $N=2$, (ii) and (iii) are empty.) \BBB Note that paths are ordered subsets of $V$ but they are not oriented, i.e., $(x_1,\ldots,x_N)$ and $(x_N,\ldots,x_1)$ should be considered as the same straight path. When taking intersections and unions, we will sometimes regard straight paths as subsets of $V$ with a slight abuse of notation. \EEE  The set of straight paths is denoted by
\begin{align*}
\Gamma (G) :=\{\gamma \text{ straight path}\}\,.
\end{align*}
 We drop $G$ and write $\Gamma$ if no confusion arises.
If $\gamma \in \Gamma$ and $x_1=x_N$, we say that $\gamma$ is \emph{closed} and otherwise that $\gamma$ is \emph{open}. In the following, we add some strata for degenerate points which will be convenient for Lemma \ref{lemma:opengraph}. Specifically,  we define 
\begin{align}\label{def:V0V1}
\begin{split}
&V_i :=\{x \in V \colon \# \mathcal{N}(x,E) = i  \} \text{  for $i=0,\ldots,4$\,,} \\&V_2^\pi :=\{x \in V_2 \colon  \theta_{x_1,x,x_2},   \in [\pi-\varepsilon,\pi+\varepsilon] \text{ where } \mathcal{N}(x,E) = \lbrace x_1,x_2 \rbrace \}\,.
\end{split}
\end{align}
 Note that in the second definition, one could equally use the angle $\theta_{x_2,x,x_1}$ as $\theta_{x_2,x,x_1} = 2\pi - \theta_{x_1,x,x_2}$.  If $x \in V_0$ we set $s(x) =\{(x),(x)\}$, if $x \in V_1 \cup V_2^\pi$ we set $s(x) = \{(x)\}$, \BBB called \emph{degenerate strata}.  We define \EEE the \emph{set of strata} by
\begin{align}\label{def:strata}
\mathcal{S}(G):=  \mathcal{S}_\Gamma \cup \bigcup_{x \in V_0 \cup V_1 \cup V_2^\pi} s(x), \quad \text{ where } \mathcal{S}_\Gamma:=\{\gamma \in \Gamma \colon \gamma \text{ is a maximal element  w.r.t.}  \subseteq \}\,.
\end{align}
\BBB We say that a stratum is \emph{open} if it is an open straight path or a degenerate stratum. Otherwise, a stratum is called \emph{closed}. \EEE
 We drop $G$ and write $\mathcal{S}$ if no confusion arises. Some closed, open, and degenerate strata are illustrated in Figure~\ref{fig:strata}. \BBB  In particular, $s(x)$ for $x \in V_0$ has to be understood as a multiset containing the stratum $(x)$ twice (strictly speaking $\mathcal{S}(G)$ is therefore the \emph{multiset} of all strata).  \EEE Adding the degenerate stratum $(x)$ with one element twice for $V_0$ and once for $V_1 \cup V_2^\pi$ has no geometrical interpretation but is merely \BBB for convenience: for graphs whose straight paths are all open, it allows us to relate the overall number of strata to $F_{\rm bond}$  and ensures that each atom is contained in exactly two strata. \EEE  More precisely, denoting by $l(s):= \# s$ the length of $s \in \mathcal{S}$, we have the following. 

\begin{figure}[htp]
\begin{tikzpicture}
\draw[ultra thin, gray] (0,-.2)--++(0,.4);
\draw[ultra thin, gray] (-.2,0)--++(.4,0);
\draw[fill=black] circle(.05);

\draw[ultra thin, gray] (2.8,0)--++(.4,0);

\draw[ultra thin, gray] (3,-2.2)--++(0,2.4);

\draw[ultra thin, gray] (2.8,0)++(10:1)++(15:1)--++(.4,0);

\draw[ultra thin, gray] (2.8,0)++(10:1)++(15:1)++(20:1)++(30:1)--++(.4,0);
\draw (3.2,0.2) node[anchor=south]{$x \in V_1$};
\draw (0.2,0.2) node[anchor=south]{$x \in V_0$};
\draw (3,-1.8)++(10:1)++(15:1)++(20:1) node[anchor=south]{$x \in V_2^\pi$};

\draw[ultra thin, gray] (3,-2.2)++(10:1)--++(0,1.4);

\draw[ultra thin, gray] (3,-2.2)++(10:1)++(15:1)--++(0,2.4);

\draw[ultra thin, gray] (3,-1.2)++(10:1)++(15:1)++(20:1)--++(0,.4);
\draw[ultra thin, gray] (3,-2.2)++(10:1)++(15:1)++(20:1)++(30:1)--++(0,2.4);

\draw[ultra thin, gray] (3,-1)--++(10:-.2)--++(10:1.2)--++(15:1)--++(20:1)--++(30:1.2);

\draw[ultra thin, gray] (3,-2)--++(10:-.2)--++(10:1.2)--++(15:1.2);
\draw[ultra thin, gray] (3,-2)++(10:1)++(15:1)++(20:1)++(30:0.8)--++(30:.4);

\draw[fill=black] (3,-1)circle(.05)++(10:1)circle(.05)++(15:1)circle(.05)++(20:1)circle(.05)++(30:1)circle(.05);

\draw[fill=black] (3,0)circle(.05)++(10:1)++(15:1)circle(.05)++(20:1)++(30:1)circle(.05);
\draw[fill=black] (3,-2)circle(.05)++(10:1)circle(.05)++(15:1)circle(.05)++(20:1)++(30:1)circle(.05);

\begin{scope}[shift={(-4,-2)},rotate=10]

\foreach \j in {0,...,2}{
\draw[ultra thin, gray] (\j,-.2)--++(0,2.4);
\draw[ultra thin, gray] (-.2,\j)--++(2.4,0);
}

\foreach \i in {0,...,2}{
\foreach \j in {0,...,2}{

\draw[fill=black](\i,\j)circle(.05);

}

}
\end{scope}

\begin{scope}[shift={(.3,-2.4)},rotate=20]

\foreach \j in {0,...,8}{

\draw[ultra thin, gray](\j*45:1.2)--(45+\j*45:1.2);

\draw[ultra thin, gray](\j*45:1)--++(\j*45:.4);

}
\foreach \j in {0,...,8}{

\draw[fill=black](\j*45:1.2) circle(.05);

}

\end{scope}
\end{tikzpicture}
\caption{Illustration of some strata in the bond graph of $G$. \BBB Degenerate strata are indicated by short segments crossing the atoms.\EEE}
\label{fig:strata}
\end{figure}
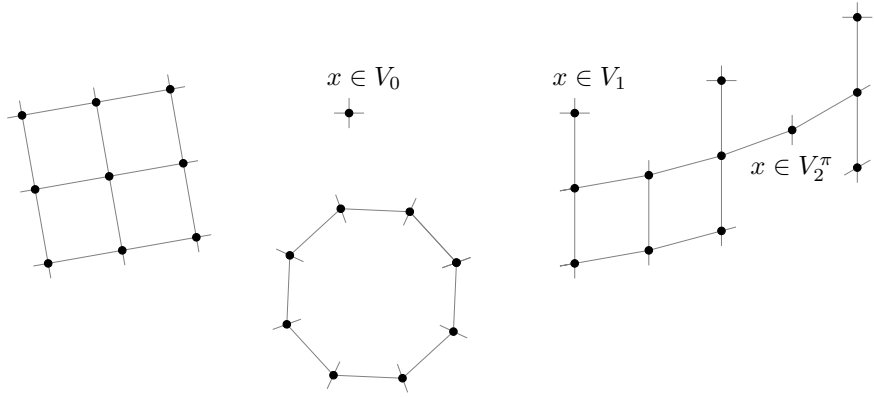

\begin{lemma}(Properties of graphs  only containing open paths)\label{lemma:opengraph}  Let $G=(V,E)$ be an $\eps$-regular graph. Assume that all $\gamma\in \Gamma$ are open. Then, the following holds: 
\begin{itemize}
\item[\rm (i)]  $\sum_{s \in \mathcal{S}} l(s) =2n$;
\item[\rm (ii)] $ F_{\rm bond}(G) = \sum_{x \in V} (4-\#\mathcal{N}(x,E)) = 2\#\mathcal{S}$.
\end{itemize}
\end{lemma}
\begin{proof} We prove the two statements in separate steps.\\
(i) \BBB It suffices to show that each $x \in V$ belongs  to exactly two $s \in \mathcal{S}$.
First, each $x \in V \setminus (V_0 \cup V_1 \cup V_2^\pi)$ lies in exactly two elements of  $\mathcal{S}_\Gamma$ and in no degenerate stratum. Indeed, as $G$ is $\varepsilon$-regular, we can find two different straight paths that contain $x$ as the only common point and whose union is not a straight path. Here, we used that $\#\mathcal{N}(x,E) \geq 2$ and $x \notin V_2^\pi$. Since all $\gamma \in \Gamma$ are open, this guarantees that there exist two different maximal straight paths containing $x$ (left figure of Figure~\ref{fig:threecases} below is excluded). The $\varepsilon$-regularity of $G$ also implies that there are at most two maximal straight paths through $x$.

  Secondly, each $x \in V_1 \cup V_2^\pi$ lies in exactly one element of $\mathcal{S}_\Gamma$ and  $x \in V_0$ is not contained in any element of $\mathcal{S}_\Gamma$. \BBB More precisely, each $x \in V_1$ is bonded to exactly one other atom and therefore forms a path. This path is contained in one maximal straight path. If $x \in V_2^\pi$, it forms a straight path together with $\mathcal{N}(x,E)$, according to definition \eqref{def:V0V1}. Again, this straight path  is contained in one maximal straight path. \EEE The definition of $s(x)$ for $x \in V_0 \cup V_1 \cup V_2^\pi$ \BBB now \EEE implies that each $x \in V$ belongs to  exactly two $s \in \mathcal{S}$. \BBB (This is the very reason for adding the degenerate strata in \eqref{def:strata}.) \EEE  Hence, {\rm (i)} follows. \\ 
(ii) We prove the statement by induction over $ m  =\#E$. It is clearly true for $m=0$ since, by definition of \eqref{def:strata}, $x \in V \implies x \in V_0$ and thus 
\begin{align*}
\#\mathcal{S} = 2\# V = \frac{1}{2}\sum_{x \in V} (4-\#\mathcal{N}(x,\emptyset)) = \frac{1}{2}\sum_{x \in V} (4-\#\mathcal{N}(x,E)) \,.
\end{align*}
 Let now $\# E=m \geq 1$ and let \BBB $s =(x_1,\ldots,x_N) \in \mathcal{S}$ \EEE be arbitrary. Consider  $\hat{E} := E \setminus \{x_1,x_2\}$ and  the corresponding graph $\hat{G}=(V,\hat{E})$. Then,
$\# \hat{E} = m-1$ and thus, by the induction hypothesis,
\begin{align*}
\sum_{x \in V} (4-
\#\mathcal{N}(x,\hat{E})) = 2\# {\mathcal{S}}(\hat{G})  \,,
\end{align*}
where ${\mathcal{S}}(\hat{G})$ is the set of strata of $\hat{G}$, defined in \eqref{def:strata}. Note that  $ {\mathcal{S}}(\hat{G})  =  (\mathcal{S} \cup \{(x_1)\} \cup \{(x_2,\ldots,x_N)\}  )\setminus s$ and thus $\# {\mathcal{S}}(\hat{G})  =  \#\mathcal{S}+1 $.  As $\#\mathcal{N}(x_i,E)= \#\mathcal{N}(x_i,\hat{E})+1$ for $i=1,2$, \EEE  we have 
\begin{align*}
\sum_{x \in V} (4-\#\mathcal{N}(x,E)) =  -2  +\sum_{x \in V} (4-
\#\mathcal{N}(x,\hat{E})) =  -2 + 2\# {\mathcal{S}}(\hat{G})  =   2 \EEE \#\mathcal{S} \,.
\end{align*}
This concludes the proof.  
\end{proof}

 We proceed with two definitions and a lemma on graphs with small angle excess. 
\begin{definition}[Angle excess] Given \BBB $\gamma=(x_1,\ldots,x_N) \in \Gamma$ for $N \geq 3$\EEE, we define the \emph{angle excess} by
\begin{align*}
\theta_{\mathrm{ex}}(\gamma) := \sum^{N-1}_{i=2} |\theta_i-\pi|\,, \quad \text{where $\theta_i =\theta_{x_{i+1},x_i,x_{i-1}}$}\,. 
\end{align*}
\BBB If $\gamma=(x_1,x_2) \in \Gamma$, we set $\theta_{\mathrm{ex}}(\gamma)=0$.
\end{definition}

\begin{definition}[Orthogonal strata]\label{def: strata-ort} Let $s \in \mathcal{S}$. We define the \emph{set of orthogonal strata} to $s$ by
\begin{align*}
\mathcal{S}^\perp(s)=\{s' \in \mathcal{S} \setminus \{s\} \colon s\cap s'\neq \emptyset\}\,.
\end{align*}
\end{definition}
 A stratum $s \in \mathcal{S}$ and its orthogonal strata are illustrated in Figure~\ref{fig:orthogonal}. \BBB For degenerate strata $s =(x) \in \mathcal{S}$ (recall the definition below \eqref{def:V0V1} and see Figure~\ref{fig:strata}), we explicitly have $\mathcal{S}^\perp(s)=\{(x)\}$ if $x \in V_0$ and $\mathcal{S}^\perp(s)=\{\gamma\}$ if $x \in V_1 
\cup V_2^\pi$, where $\gamma\in \mathcal{S}_\Gamma$ is the unique maximal straight path containing $x$, cf.\ proof of Lemma \ref{lemma:opengraph}(i).  The next lemma shows some elementary properties of graphs with small angle excess.  

\EEE

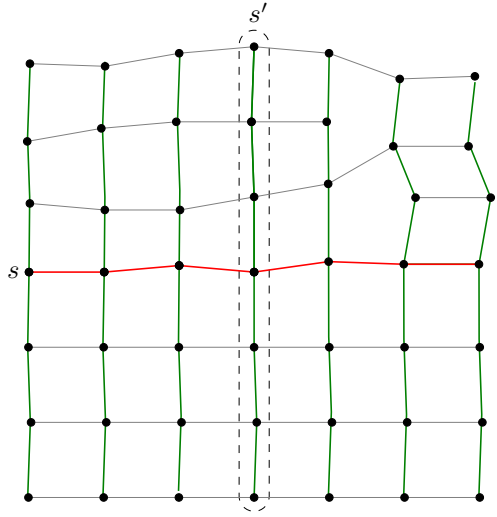
\begin{figure}
\begin{tikzpicture}

\draw(0,0)++(0:1.05)++(5:1)++(-5:1)++(89.5:1)++(92:.9)++(88:.9)++(90:.4) node[anchor=south]{$s'$};

\draw[ultra thin,dashed](0,0)++(0:1)++(5:1)++(-5:.8)--++(90:-3) arc(180:360:.2)--++(90:6) arc (0:180:.2)--++(270:3);

\draw[semithick,green!50!black](0,0)++(0:1)++(5:1)++(-5:1)++(8:1)++(-2:1)++(0:1)--++(90.1:-1)--++(92:-1)--++(88:-1.1);

\draw[semithick,green!50!black](0,0)++(0:1)++(5:1)++(-5:1)++(8:1)++(-2:1)--++(90.1:-1)--++(92:-1)--++(88:-1.1);

\draw[semithick,green!50!black](0,0)++(0:1)++(5:1)++(-5:1)++(8:1)--++(90.1:-1)--++(92:-1)--++(88:-1.1);

\draw[semithick,green!50!black](0,0)++(0:1)++(5:1)--++(89.5:-1)--++(92:-1)--++(88:-1);

\draw[semithick,green!50!black](0,0)++(0:1)--++(89:-1)--++(92:-1)--++(88:-1);

\draw[semithick,green!50!black](0,0)--++(89:-1)--++(92:-1)--++(88:-1);

\draw[ultra thin,gray](0,0)++(0:1)++(5:1)++(-5:1)++(90:-1)++(92:-1)++(88:-1)--++(180:1)--++(180:1)--++(180:1);

\draw[ultra thin,gray](0,0)++(0:1)++(5:1)++(-5:1)++(90:-1)++(92:-1)++(88:-1)--++(0:1)--++(0:1)--++(0:1);

\draw[ultra thin,gray](0,0)++(0:1)++(5:1)++(-5:1)++(90:-1)++(92:-1)--++(180:1)--++(180:1)--++(180:1);

\draw[ultra thin,gray](0,0)++(0:1)++(5:1)++(-5:1)++(90:-1)++(92:-1)--++(0:1)--++(0:1)--++(0:1);

\draw[ultra thin,gray](0,0)++(0:1)++(5:1)++(-5:1)++(90:-1)--++(180:1)--++(180:1)--++(180:1);

\draw[ultra thin,gray](0,0)++(0:1)++(5:1)++(-5:1)++(90:-1)--++(0:1)--++(0:1)--++(0:1);

\draw[fill=black](0,0)++(0:1)++(5:1)++(-5:1)++(90:-1)++(180:1)circle(.05)++(180:1)circle(.05)++(180:1)circle(.05);

\draw[fill=black](0,0)++(0:1)++(5:1)++(-5:1)++(90:-1)++(0:1)circle(.05)++(0:1)circle(.05)++(0:1)circle(.05);

\draw[fill=black](0,0)++(0:1)++(5:1)++(-5:1)++(90:-1)++(92:-1)++(180:1)circle(.05)++(180:1)circle(.05)++(180:1)circle(.05);

\draw[fill=black](0,0)++(0:1)++(5:1)++(-5:1)++(90:-1)++(92:-1)++(0:1)circle(.05)++(0:1)circle(.05)++(0:1)circle(.05);

\draw[fill=black](0,0)++(0:1)++(5:1)++(-5:1)++(90:-1)++(92:-1)++(88:-1)++(180:1)circle(.05)++(180:1)circle(.05)++(180:1)circle(.05);

\draw[fill=black](0,0)++(0:1)++(5:1)++(-5:1)++(90:-1)++(92:-1)++(88:-1)++(0:1)circle(.05)++(0:1)circle(.05)++(0:1)circle(.05);

\draw[semithick,green!50!black](0,0)++(0:1)++(5:1)++(-5:1)++(8:1)++(-2:1)--++(0:1)--++(80:.9)--++(112:.8)--++(83.5:.8);

\draw[ultra thin,gray](0,0)++(0:1)++(5:1)++(-5:1)++(8:1)++(-2:1)++(80:.9)--++(0:1);

\draw[semithick,green!50!black](0,0)++(0:1)++(5:1)++(-5:1)++(8:1)++(-2:1)--++(80:.9)--++(112:.8)--++(83.5:.8);

\draw[semithick,green!50!black](0,0)++(0:1)++(5:1)++(-5:1)++(8:1)--++(90:1.1)--++(90.5:.8)--++(89:.9);
\draw[semithick,green!50!black](0,0)--++(89.5:1)--++(92:.8)--++(88:1);

\draw[semithick,green!50!black](0,0)++(0:1)--++(89.5:1)--++(92:.9)--++(88:.9);

\draw[semithick,green!50!black](0,0)++(0:1)++(5:1)--++(89.5:1)--++(92:.9)--++(88:.9);


\draw[semithick,red](0,0)--++(0:1)--++(5:1)--++(-5:1)--++(8:1)--++(-2:1)--++(0:1);

\draw[ultra thin,gray](0,0)++(0:1)++(5:1)++(-5:1)++(90:1)--++(10:1)--++(30:1)--++(0:1);

\draw[ultra thin,gray](0,0)++(0:1)++(5:1)++(-5:1)++(90:1)++(92:1)--++(0:1);

\draw[ultra thin,gray](0,0)++(0:1)++(5:1)++(-5:1)++(90:1)++(92:1)--++(0:-1)--++(5:-1)--++(10:-1);

\draw[ultra thin,gray](0,0)++(0:1)++(5:1)++(-5:1)++(90:1)++(92:1)++(88:1)--++(-5:1)--++(-20:1)--++(2:1);

\draw[ultra thin,gray](0,0)++(0:1)++(5:1)++(-5:1)++(90:1)++(92:1)++(88:1)--++(5:-1)--++(10:-1)--++(-2:-1);

\draw[fill=black](0,0)++(0:1)++(5:1)++(-5:1)++(90:1)++(92:1)++(88:1)++(5:-1)circle(.05)++(10:-1)circle(.05)++(-2:-1)circle(.05);

\draw[fill=black](0,0)++(0:1)++(5:1)++(-5:1)++(8:1)++(-2:1)++(0:1)++(80:.9)circle(.05);

\draw[fill=black](0,0)++(0:1)++(5:1)++(-5:1)++(90:1)++(92:1)++(88:1)++(-5:1)circle(.05)++(-20:1)circle(.05)++(2:1)circle(.05);

\draw[ultra thin,gray](0,0)++(0:1)++(5:1)++(-5:1)++(90:1)--++(10:-1)--++(0:-1)--++(-5:-1);

\draw[semithick,green!50!black](0,0)++(0:1)++(5:1)++(-5:1)--++(90:1)--++(92:1)--++(88:1);

\draw[semithick,green!50!black](0,0)++(0:1)++(5:1)++(-5:1)--++(90:1)--++(92:1)--++(88:1);
\draw[semithick,green!50!black](0,0)++(0:1)++(5:1)++(-5:1)--++(90:-1)--++(92:-1)--++(88:-1);

\draw[fill=black](0,0)++(0:1)++(5:1)++(-5:1)++(90:1)++(92:1)++(0:1)circle(.05);

\draw[fill=black](0,0)++(0:1)++(5:1)++(-5:1)++(90:1)++(92:1)circle(.05)++(0:-1)circle(.05)++(5:-1)circle(.05)++(10:-1)circle(.05);

\draw[fill=black](0,0)++(0:1)++(5:1)++(-5:1)++(8:1)++(-2:1)++(80:.9)circle(.05);

\draw[fill=black](0,0)++(0:1)++(5:1)++(-5:1)++(90:1)++(10:-1)circle(.05)++(0:-1)circle(.05)++(-5:-1)circle(.05);

\draw[fill=black](0,0)++(0:1)++(5:1)++(-5:1)++(90:1)++(10:1)circle(.05)++(30:1)circle(.05)++(0:1)circle(.05);

\draw[fill=black](0,0)++(0:1)circle(.05)++(5:1)circle(.05)++(-5:1)circle(.05)++(90:1)circle(.05)++(92:1)circle(.05)++(88:1)circle(.05);

\draw[fill=black](0,0)circle(.05)++(0:1)circle(.05)++(5:1)circle(.05)++(-5:1)circle(.05)++(90:-1)circle(.05)++(92:-1)circle(.05)++(88:-1)circle(.05);

\draw[fill=black](0,0)circle(.05)++(0:1)circle(.05)++(5:1)circle(.05)++(-5:1)circle(.05)++(8:1)circle(.05)++(-2:1)circle(.05)++(0:1)circle(.05);

\draw(0,0) node[anchor=east]{$s$};

\end{tikzpicture}
\caption{The stratum $s$, in red, and its orthogonal strata $\mathcal{S}^\perp(s)$ in green. One $s'\in \mathcal{S}^\perp(s)$ is encircled.}
\label{fig:orthogonal}
\end{figure}

\begin{lemma}(Small angle excess \BBB for regular graphs\EEE)\label{lemma:excess} Let $G=(V,E)$ be an $\eps$-regular graph. The following implications hold true:
\begin{itemize}
\item[\rm (i)] If  $\displaystyle\max_{\gamma \in \Gamma} \theta_{\mathrm{ex}}(\gamma) < \frac{3\pi}{2}-\eps   $, then all \BBB $\gamma \in \Gamma$ are open\EEE;
\item[\rm (ii)]  If $\displaystyle \max_{\gamma \in \Gamma} \theta_{\mathrm{ex}}(\gamma) < \frac{\pi}{2}-\varepsilon$, then $\#\mathcal{S}^\perp(s) = l(s)$ for all $s \in \mathcal{S}$;
\item[\rm (iii)] If $\displaystyle \max_{\gamma \in \Gamma} \theta_{\mathrm{ex}}(\gamma) < \frac{\pi}{6}-\varepsilon$, then $s_1 \cap s_2= \emptyset$ for all $s_1,s_2 \in \mathcal{S}^\perp(s)$ and for all $s \in \mathcal{S}$.
\end{itemize}
\end{lemma}

\begin{proof}  We first introduce some notation that will be used throughout the proof.   Let $ p = \{x_1,\ldots,x_N\}$ be such that the edges $e_i=\{x_i,x_{i+1}\}$, $i=1,\ldots,N-1$,  form  a closed \BBB simple \EEE polygon.  We denote by $\theta(e_i,e_{i+1})$ the interior angle formed by the edges $e_i  $ and $e_{i+1}$, $i=1,\ldots,N-1$,   with the convention    $e_{1}=e_N$. By the interior angle sum of polygons it holds 
\begin{align}\label{eq:intsum}
\sum_{i=2}^{ N} ( \theta(e_{i-1},e_{i})  -\pi) =-2\pi\,.
\end{align}
 For the reader's convenience, the proof of the three different statements is aided by Figure~\ref{fig:threecases}. \EEE
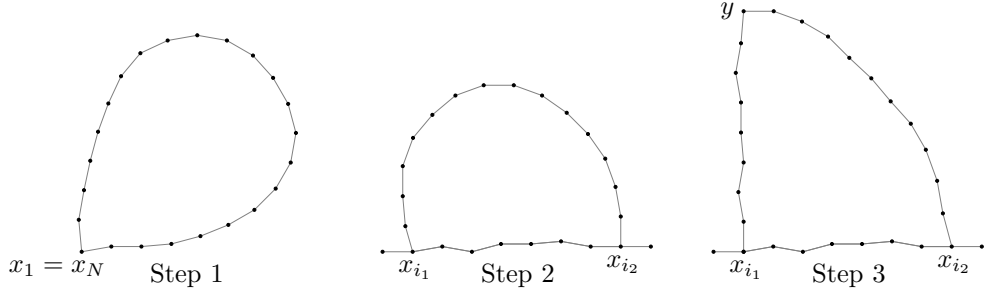
\begin{figure}[htp]
\begin{tikzpicture}[scale=.4]

\draw(-.8,0) node[anchor=north]{$x_1=x_N$};

\draw[ultra thin, gray](0,0)--++(10:1)--++(0:1)--++(5:1)--++(15:1)--++(22:1)--++(30:1)--++(45:1)--++(60:1)--++(80:1)--++(105:1)--++(120:1)--++(132:1)--++(150:1)--++(170:1)--++(190:1)--++(205:1)--++(230:1)--++(245:1)--++(250:1)--++(255:1)--++(258:1)--++(260:1)--(0,0);

\draw[fill=black](0,0)circle(.05)++(10:1)circle(.05)++(0:1)circle(.05)++(5:1)circle(.05)++(15:1)circle(.05)++(22:1)circle(.05)++(30:1)circle(.05)++(45:1)circle(.05)++(60:1)circle(.05)++(80:1)circle(.05)++(105:1)circle(.05)++(120:1)circle(.05)++(132:1)circle(.05)++(150:1)circle(.05)++(170:1)circle(.05)++(190:1)circle(.05)++(205:1)circle(.05)++(230:1)circle(.05)++(245:1)circle(.05)++(250:1)circle(.05)++(255:1)circle(.05)++(258:1)circle(.05)++(260:1)circle(.05);

\draw(3.5,-1.5) node[anchor=south]{Step 1};

\begin{scope}[shift={(11,0)}]

\draw[ultra thin,gray](0,0)--++(180:1);

\draw[ultra thin, gray](0,0)--++(10:1)--++(-10:1)--++(15:1)--++(0:1)--++(5:1)--++(-10:1)--++(0:1)--++(0:1);

\draw[fill=black](0,0)++(180:1) circle(.05);

\draw[fill=black](0,0)++(10:1)++(-10:1)++(15:1)++(0:1)++(5:1)++(-10:1)++(0:1)++(0:1) circle(.05);

\draw(.1,0) node[anchor=north]{$x_{i_1}$};

\draw(0,0)++(10:1.1)++(-10:1)++(15:1)++(0:1)++(5:1)++(-10:1)++(0:1) node[anchor=north]{$x_{i_2}$};

\draw(3.5,-1.5) node[anchor=south]{Step 2};

\draw[ultra thin, gray](0,0)--++(10:1)--++(-10:1)--++(15:1)--++(0:1)--++(5:1)--++(-10:1)--++(0:1)--++(90:1)--++(100:1)--++(110:1)--++(125:1)--++(135:1)--++(145:1)--++(160:1)--++(180:1)--++(200:1)--++(220:1)--++(230:1)--++(250:1)--++(270:1)--++(275:1)--(0,0);

\draw[fill=black](0,0)circle(.05)++(10:1)circle(.05)++(-10:1)circle(.05)++(15:1)circle(.05)++(0:1)circle(.05)++(5:1)circle(.05)++(-10:1)circle(.05)++(0:1)circle(.05)++(90:1)circle(.05)++(100:1)circle(.05)++(110:1)circle(.05)++(125:1)circle(.05)++(135:1)circle(.05)++(145:1)circle(.05)++(160:1)circle(.05)++(180:1)circle(.05)++(200:1)circle(.05)++(220:1)circle(.05)++(230:1)circle(.05)++(250:1)circle(.05)++(270:1)circle(.05)++(275:1)circle(.05)(0,0);;

\end{scope}

\begin{scope}[shift={(22,0)}]

\draw[ultra thin,gray](0,0)--++(180:1);

\draw[ultra thin, gray](0,0)--++(10:1)--++(-10:1)--++(15:1)--++(0:1)--++(5:1)--++(-10:1)--++(0:1)--++(0:1);

\draw[fill=black](0,0)++(180:1) circle(.05);

\draw[fill=black](0,0)++(10:1)++(-10:1)++(15:1)++(0:1)++(5:1)++(-10:1)++(0:1)++(0:1) circle(.05);

\draw(0,8) node[anchor=east]{$y$};

\draw(.1,0) node[anchor=north]{$x_{i_1}$};

\draw(0,0)++(10:1.1)++(-10:1)++(15:1)++(0:1)++(5:1)++(-10:1)++(0:1) node[anchor=north]{$x_{i_2}$};

\draw(3.5,-1.5) node[anchor=south]{Step 3};

\draw[ultra thin, gray](0,0)--++(10:1)--++(-10:1)--++(15:1)--++(0:1)--++(5:1)--++(-10:1)--++(0:1);

\draw[ultra thin, gray](0,0)--++(90:1)--++(100:1)--++(80:1)--++(95:1)--++(90:1)--++(100:1)--++(80:1)--(0,8);

\draw[ultra thin, gray](0,8)--++(0:1)--++(-20:1)--++(-30:1)--++(-45:1)--++(-43:1)--++(-50:1)--++(-48:1)--++(-60:1)--++(-70:1.1)--++(-80:1.1)--++(-75:1.1);

\draw[fill=black](0,8)++(0:1)circle(.05)++(-20:1)circle(.05)++(-30:1)circle(.05)++(-45:1)circle(.05)++(-43:1)circle(.05)++(-50:1)circle(.05)++(-48:1)circle(.05)++(-60:1)circle(.05)++(-70:1.1)circle(.05)++(-80:1.1)circle(.05)++(-75:1.1);

\draw[fill=black](0,8)circle(.05);

\draw[fill=black](0,0)circle(.05)++(90:1)circle(.05)++(100:1)circle(.05)++(80:1)circle(.05)++(95:1)circle(.05)++(90:1)circle(.05)++(100:1)circle(.05)++(80:1)circle(.05);

\draw[fill=black](0,0)circle(.05)++(10:1)circle(.05)++(-10:1)circle(.05)++(15:1)circle(.05)++(0:1)circle(.05)++(5:1)circle(.05)++(-10:1)circle(.05)++(0:1)circle(.05);

\end{scope}

\end{tikzpicture}
\caption{The three cases discussed in Steps~1--3.}
\label{fig:threecases}
\end{figure}
\begin{step}{1}(Proof of (i))   Assume by contradiction that $\theta_{\mathrm{ex}}(\gamma) < \frac{3\pi}{2}-\eps  $ for all $\gamma \in \Gamma$ and that there exists $\gamma \in \Gamma$ closed. Let \BBB $\gamma =(x_1,\ldots,x_N)  \in \Gamma$ \EEE be a minimal (w.r.t.\ set inclusion)  closed path. \BBB Since the graph is planar,  the edges $e_i=\{x_i,x_{i+1}\}$, $i=1,\ldots,N-1$, form  a closed simple polygon. \EEE Therefore, by \eqref{eq:intsum} and the triangle inequality, we have $\theta_{\mathrm{ex}}(\gamma) + |\theta(e_{N-1}, e_1)  - \pi| \geq 2\pi$. Since $|\theta(e_{N-1}, e_1) - \pi| \le \frac{\pi}{2}+\eps$ by Definition~\ref{def:epsregular}(ii), this  yields a contradiction and concludes Step 1.
\end{step}\\
\begin{step}{2}(Proof of (ii))   Assume by contradiction that $\theta_{\mathrm{ex}}(\gamma) < \frac{\pi}{2}-\varepsilon$ for all $\gamma \in \Gamma$ and that there exists \BBB $s=(x_1,\ldots,x_N) \in \mathcal{S}$ \EEE with $\#\mathcal{S}^\perp(s) <l(s)$. This implies that $N \geq 2$.  Moreover, there exists $s'\in \mathcal{S}^\perp(s)$ and $1\leq i_1<i_2\leq N$ such that $ \{x_{i_1},x_{i_2}\} \subset s'\cap s $. Let us consider $\gamma \subset s'$ connecting $x_{i_1}$ and $x_{i_2}$ such that $\gamma \cap s =  \{x_{i_1},x_{i_2}\} $. We now consider \BBB $p =(y_1,\ldots,y_M) = \phi \cup \gamma$, where $\phi:=(x_{i_1},\ldots, x_{i_2}) \subset s$\EEE,   and observe that its edges $e_i$, $i=1,\ldots,M-1$,  form a closed polygon. Let $y_j = x_{i_1}$ and $y_k = x_{i_2}$. Note that   $|\theta(e_{j-1},e_j)-\pi|,|\theta(e_{k-1},e_k)-\pi|\leq \frac{\pi}{2}+\varepsilon$ by Definition~\ref{def:epsregular}(ii).  Identity \eqref{eq:intsum}   applied to $p$ implies  $\sum_{i=2}^{ M } ( \theta(e_{i-1},e_{i})-\pi) =-2\pi$.  Furthermore,  $\phi,\gamma \in \Gamma$ and therefore  we obtain
\begin{align*}
2\pi=\left|\sum_{i=2}^M( \theta(e_{i-1},e_{i})  -\pi)\right| &\leq  | \theta(e_{j-1},e_j) -\pi| + | \theta(e_{k-1},e_k)-\pi | +  \underset{ i\neq \{j,k\}}{\sum_{i=2}^M}|  \theta(e_{i-1},e_{i}) -\pi| \\&\leq \pi +2\varepsilon +  \underset{i\neq \{j,k\}}{\sum_{i=2}^M}|  \theta(e_{i-1},e_{i})  -\pi| = \pi +2\varepsilon + \theta_{\mathrm{ex}}(\gamma) + \theta_{\mathrm{ex}}(\phi)  \,.
\end{align*}
This implies  that   $\theta_{\mathrm{ex}}(\gamma) \geq \pi/2-\varepsilon$ or $\theta_{\mathrm{ex}}(\phi) \geq \pi/2-\varepsilon$ and yields therefore a contradiction.
\end{step}\\
\begin{step}{3}(Proof of (iii)) Assume by contradiction that $\theta_{\mathrm{ex}}(\gamma) < \pi/6 -\varepsilon $ for all $\gamma \in \Gamma$ and  that there exists $s\in \mathcal{S}$ and $s_1,s_2 \in \mathcal{S}^\perp(s)$ such that $s_1 \cap s_2 \neq \emptyset$. Writing \BBB $s=(x_1,\ldots,x_N)$ \EEE there exists $i_1<i_2$ such that $s_1 \cap s = \{x_{i_1}\}$ and  $s_2 \cap s = \{x_{i_2}\}$. (Due to Step 2, there is only one point of intersection between $s_i$ and $s$.) Again by Step 2, there holds $\{y\} = s_1 \cap s_2$ for some $y \in V$. Denote by \BBB $\gamma_1 =(y_1,\ldots,y_{l_1}) \subset s$ \EEE  the path connecting $x_{i_1}$ with $x_{i_2}$, \BBB $\gamma_2= (y_{l_1},\ldots,y_{l_2}) \subset s_2$ \EEE the path connecting $x_{i_2}$ with $y$, and by \BBB $\gamma_3= (y_{l_2},\ldots,y_{l_3}) \subset s_1$ \EEE the path connecting $y$ with $x_{i_1}$.  Note that $\gamma_i \in \Gamma$ for $i=1,2,3$. We set $p =\gamma_1 \cup \gamma_2 \cup \gamma_3$ and observe that the edges of $p$ form a closed polygon. Equation \eqref{eq:intsum} applied for $p$ implies $\sum_{i=2}^{l_3}( \theta(e_{i-1},e_{i}) -\pi) =-2\pi$. Note that $| \theta(e_{l_k-1},e_{l_k}) - \pi| \le \frac{\pi}{2}+\varepsilon$  for all $k=1,2,3$ by Definition~\ref{def:epsregular}(ii).  Therefore, we obtain
\begin{align*}
2\pi = \left|\sum_{i=2}^{l_3}( \theta(e_{i-1},e_{i}) -\pi)\right|  &\leq \sum_{k=1}^3 | \theta(e_{l_k-1},e_{l_k}) -\pi| +\underset{i\neq \{l_1,l_2,l_3\}}{\sum_{i=2}^{l_3}}\left| \theta(e_{i-1},e_{i}) -\pi\right|\\&\leq \frac{3\pi}{2} +3\varepsilon + \sum_{k=1}^3 \theta_{\mathrm{ex}}(\gamma_k)\,.
\end{align*}
This implies that there exists $k\in \{1,2,3\}$ such that $\theta_{\mathrm{ex}}(\gamma_k) \geq \pi/6-\varepsilon$, a contradiction. 
\end{step}
\end{proof}

 We now come to the \textit{stratification} of bond graphs. The following lemma   allows to reduce the problem of crystallization to a purely geometric problem of minimizing the number of strata in graphs containing only open strata  with small angle excess. \BBB This is the only step in the proof where {($\rm iii_3$)} is needed. \EEE

\begin{lemma}(Construction of a \BBB regular \EEE graph with small angle excess)\label{lemma:construction} Let $ G=(V,E)$ be $\varepsilon$-regular. Then, there exists $G_\mathrm{o}= (V,E_\mathrm{o})$ with $ E_\mathrm{o} \subset E$ such that
\begin{itemize}
\item[\rm (i)] $\max_{\gamma \in \Gamma(G_\mathrm{o})} \theta_{\mathrm{ex}}(\gamma) < \frac{\pi}{6}-\varepsilon$\,;
\item[\rm (ii)]  $G_{\mathrm{o}}$ satisfies
\begin{align*}
F(G) \geq  F_{\mathrm{bond}}(  G_{\mathrm{o}})  \EEE \,
\end{align*}
with equality only if $E = E_\mathrm{o}$,  $|x-y|=1$ for all $x \in V$, $y\in \mathcal{N}(x,E)$, and $\theta \in \{\pi/2,\pi,3\pi/2\}$ for all bond angles $\theta$.
\end{itemize}
\end{lemma} 
\begin{proof} We construct $G_\mathrm{o}= (V,E_\mathrm{o})$ by iteratively erasing edges. We start by setting $G^0= (V,E)$ and we suppose that $G^k=(V,E^k)$ is already given. We construct $G^{k+1}=(V,E^{k+1})$ by suitably modifying the set of edges $E^k$. If (i) is satisfied, we may stop. Thus, we assume that there exists  $\gamma \in \Gamma(G_k)$ such that $\theta_{\mathrm{ex}}(\gamma) \geq \pi/6-\varepsilon$. Let $\gamma_k \in \Gamma(G_k)$   be minimal (w.r.t.\ set inclusion) such that $\theta_{\mathrm{ex}}(\gamma_k) \geq \pi/6-\varepsilon$, i.e., \BBB $\gamma_k = (x_1,x_2,\ldots,x_{N-1},x_N)$ \EEE and \BBB $\hat{\gamma}_k := (x_2,\ldots,x_{N-1})$ \EEE satisfies $\theta_{\mathrm{ex}}(\hat{\gamma}_k) <\pi/6-\varepsilon$. We define $E^{k+1} := E^k\setminus (\{x_1,x_2\} \cup \{x_{N-1},x_N\})$ and  $G^{k+1}=(V,E^{k+1})$.   Then,
\begin{align}\label{eq:bonds}
\sum_{x \in V} (4-\#\mathcal{N}(x,E^{k+1})) =4+\sum_{x \in V} (4-\#\mathcal{N}(x,E^{k}))\,.
\end{align}
Additionally, due to {($\rm iii_3$)}, with $ L:= 4(\pi/6-\varepsilon)^{-1} >0$ we have
\begin{align}\label{ineq:gamma}
\sum_{ j=2}^{N-1} v_3(\theta_j)   >   L \sum_{j=2}^{N-1}|\theta_j-\pi| \ge L(\pi/6-\varepsilon) =  4\,.
\end{align}
Therefore, due to \eqref{eq:bonds} and \eqref{ineq:gamma}, we have that
\begin{align}\label{ineq:stepk}
F_{\mathrm{ex}}(\gamma_k) + F_{\mathrm{bond}}(G^k) > \BBB F_\mathrm{bond}\EEE(G^{k+1})\,,
\end{align}
where $F_{\mathrm{ex}}(\gamma_k) $ is defined in \eqref{def:Felasticlocal}. 
Since $G=(V,E)$ is finite, the procedure terminates for some $K \in \mathbb{N}$ and we set $G_\mathrm{o}:=(V,E^{K})$. By construction, $G_\mathrm{o}$ satisfies (i). It remains to show (ii). 
Note that, due to the minimal selection of $\gamma_k \in \Gamma$, once $\gamma_k$ is selected this way, we will not select any $\gamma'\subset \gamma_k$ in any successive step  $j >k$. Thus, using \eqref{ineq:stepk} and the previous observation, we have
\begin{align}\label{ineq:K}
\begin{split}
F_{\mathrm{ex}}(G) + F_{\mathrm{bond}}(G) &=F_{\mathrm{ex}}(G^0) + \BBB F_{\mathrm{bond}}\EEE(G^0)   \ge  \sum_{k=0}^{K-1}\EEE F_{\mathrm{ex}}(\gamma_k) + \BBB F_{\mathrm{bond}}\EEE(G^0)   \\& \geq F_{\mathrm{bond}}(G^K)= F_{\mathrm{bond}}(G_\mathrm{o}) \,
\end{split}
\end{align}
with strict inequality whenever $K \ge 1$. In particular,  if equality holds in \eqref{ineq:K}, we  have that $G_\mathrm{o} =G$. This  necessarily gives $F_{\mathrm{ex}}(G)=0$ which implies that $|x-y|=1$ for all $x \in V, y \in \mathcal{N}(x,E)$, and $\theta \in \{\pi/2,\pi,3\pi/2\}$ for all bond angles by  {\rm ($\rm{i}_2$)} and  {\rm ($\rm{ii}_3$)}.  This concludes the proof.
\end{proof}

\section{Proof of the main results}\label{sec: main}

 This section is devoted to the proofs of Theorems \ref{thm:crystallization}--\ref{th: 34law}.

\subsection{Crystallization}  
We will show that the minimum of $F$ is given by $2\lceil 2\sqrt{n}\rceil$, and that it is attained by subsets of $\Z^2$. In view of \eqref{eq: relation}, this shows Theorem \ref{thm:crystallization}. Recall the definition of $G_{\mathrm{nat}}$ in \eqref{eq: relation-new}.  We first state the following upper bound.

\begin{lemma}(Upper bound) \label{lemma:energybound}Let $C_n$ be a minimizer of \eqref{eq: main energy}. Then, $G_{\mathrm{nat}}$ satisfies
\begin{align*}
F(G_{\mathrm{nat}}) \leq 2\lceil 2\sqrt{n}\rceil\,.
\end{align*}
\end{lemma}
\begin{proof}
The statement is obtained  by direct construction of configurations $C_n$ with $C_n\subset \mathbb{Z}^2 $ satisfying the energy bound. We refer to \cite[Section 4]{Mainini-Piovano} for details, see also Figure \ref{fig:groundstates}.  
\end{proof}

 The core of the proof now consists in proving a lower bound. 

\begin{proof}[Proof of Theorem \ref{thm:crystallization}] Let $C_n$ be a minimizer of \eqref{eq: main energy}.  Then $G_{\mathrm{nat}}$ is $\eps$-regular by Lemma \ref{lemma:elementaryprop}.  We denote by $G_\mathrm{o}=(V,E_\mathrm{o})$ the graph obtained in Lemma \ref{lemma:construction}. The graph $G_\mathrm{o}$ is also $\varepsilon$-regular and satisfies
\begin{align}\label{ineq:prop}
\max_{\gamma \in \Gamma(G_o)}\theta_{\mathrm{ex}}(\gamma) < \pi/6-\varepsilon\,. 
\end{align} 
 The main part of the proof consists in verifying  
\begin{align}\label{eq: main claim}
F_{\mathrm{bond}}(  G_{\mathrm{o}}  ) \geq 2\lceil 2\sqrt{n}\rceil.
\end{align}
 Once  \eqref{eq: main claim} is proven,  we conclude as follows. First, \eqref{eq: main claim0} holds due to Lemma \ref{lemma:construction}(ii)  and \eqref{eq: relation}. To characterize the equality case, we get from  Lemma \ref{lemma:construction} that  $G=G_\mathrm{o}$, that all bond lengths are $1$, and all bond angles lie in $\{\pi/2,\pi,3/2\pi\}$. This shows that each connected component (in the sense of graph theory) of $G$ lies in a rotated and shifted version of $\Z^2$. If there existed more than one connected component, one could obtain a modified configuration with an additional bond. This contradicts minimality, and we therefore obtain $V \subset \mathbb{Z}^2$ after a rigid motion.

We now show  \eqref{eq: main claim}. In the following, we write $\mathcal{S}$ in place of $\mathcal{S}(G_{\mathrm{o}})$ for simplicity.  By  Lemma~\ref{lemma:energybound}, Lemma \ref{lemma:opengraph}(ii), and Lemma \ref{lemma:construction} we have that 
\begin{align}\label{ineq:cardstrata}
2\#\mathcal{S} = F_{\mathrm{bond}}( G_{\mathrm{o}} ) \leq 2\lceil 2\sqrt{n}\rceil\,.
\end{align}
Furthermore, by Lemma \ref{lemma:opengraph}(i) we have
\begin{align*}
\sum_{s \in \mathcal{S}}l(s) = 2n \,.
\end{align*}
Hence, there exists $s_0 \in \mathcal{S}$ such that
\begin{align}\label{eq: lengths_0}
l(s_0) = \max_{s \in \mathcal{S}} l(s) \geq \frac{1}{\#\mathcal{S}} \sum_{s \in \mathcal{S}} l(s) \geq  \frac{2n}{\lceil 2\sqrt{n}\rceil}\,.
\end{align}
 Recall Definition \ref{def: strata-ort} and define   $l^v:=\max_{s \in \mathcal{S}^\perp(s_0)}l(s)$ and $s^v \in \mathrm{argmax}_{s \in \mathcal{S}^\perp(s_0)} l(s)$. We claim that
\begin{align}\label{ineq:lowerboundnum}
\#(\mathcal{S}^\perp(s_0) \cup \mathcal{S}^\perp(s^v))= \#\mathcal{S}^\perp(s_0) + \#\mathcal{S}^\perp(s^v)  = l(s_0) + l^v .
\end{align}
In fact, by \eqref{ineq:prop} and Lemma~\ref{lemma:excess}(ii) we have that $\#\mathcal{S}^\perp(s_0) = l(s_0)$, $\#\mathcal{S}^\perp(s^v) = l^v$ and, by   Lemma~\ref{lemma:excess}(iii), if $s \in \mathcal{S}^\perp(s_0)$, then $s \notin \mathcal{S}^\perp(s^v)$ (and vice versa). This yields \eqref{ineq:lowerboundnum}. We set $\mathrm{span}(s_0) = \bigcup_{s'\in \mathcal{S}^\perp(s_0)}  s'  \subseteq V$ and we consider two cases:
\begin{itemize}
\item[(a)] $\mathrm{span}(s_0) = V$;
\item[(b)] $\mathrm{span}(s_0) \subsetneq  V$;
\end{itemize}
\noindent \emph{Proof in case (a):} Due to \eqref{ineq:lowerboundnum}, we get
\begin{align*}
\#\mathcal{S} \geq \#(\mathcal{S}^\perp(s_0) \cup \mathcal{S}^\perp(s_v))  = l^v +l(s_0)\,.
\end{align*}
 Now, since $\mathrm{span}(s_0)=V$, we have in particular that  $l(s_0) \cdot l^v \geq n$ and therefore, noting that $\lceil t \rceil < t+1$ we obtain    by Lemma \ref{lemma:opengraph}(ii) \BBB and Young's inequality \EEE
\begin{align}\label{eq: slice}
F_{\mathrm{bond}}( G_{\mathrm{o}}  )  =  2\#\mathcal{S} \geq 2(l(s_0) +l^v) \geq 2\left(l(s_0) + \frac{n}{l(s_0)}\right) \geq 4\sqrt{n} > 2\lceil 2\sqrt{n}\rceil-2\,.
\end{align}
 Since   $2\#\mathcal{S} \in 2\mathbb{N}$, the previous estimate yields the claim \eqref{eq: main claim} in case (a). \BBB Before proceeding with case (b), we would like to mention that the proof of inequality \eqref{eq: slice} is the analogous step to the continuum isoperimetric inequality proved in Theorem \ref{thm: slic}. \EEE \\
\noindent \emph{Proof in case (b):} We claim that in this case we have that 
\begin{align}\label{ineq:cardestimateb}
 \#\mathcal{S}  \ge  l^v + l(s_0)+1\,.
\end{align}
In fact, by definition, there exists $x \in V \setminus \mathrm{span}(s_0)$. Due to Lemma \ref{lemma:excess}(iii), for $s,s'\in \mathcal{S}^\perp(s^v)$ we have that $s \cap s'=\emptyset$,  and thus   there exists at most one $s \in S^\perp(s^v)$ such that $s\cap \{x\} \neq \emptyset$.  We also note that $s'\cap \{x\} = \emptyset$ for all $s'\in \mathcal{S}^\perp(s_0)$.  Since for all $x \in V$ there exist two strata $s, s'$ such that $x\in s, s'$  (see proof of Lemma \ref{lemma:opengraph}(i)),  there exists at least one stratum $s \notin \mathcal{S}^\perp(s_0) \cup S^\perp(s^v)$. Therefore \eqref{ineq:cardestimateb} follows. 

We denote by $\mathcal{S}^\mathrm{a} := \mathcal{S} \setminus (\mathcal{S}^\perp(s_0) \cup \mathcal{S}^\perp(s^v))$ and observe that by \eqref{ineq:cardstrata} and \eqref{ineq:lowerboundnum}  it holds that
\begin{align*}
\#\mathcal{S}^\mathrm{a} = \#\mathcal{S} - l^v-l(s_0)  \le \lceil 2\sqrt{n}\rceil - l^v-l(s_0)\,.
\end{align*}
Now, by Lemma \ref{lemma:excess}(ii) and the choice of $s_0$, see \eqref{eq: lengths_0}, and $s^v$ respectively, we have  $\# \mathcal{S}^\perp(s_0) = l(s_0)$,  $\# \mathcal{S}^\perp(s^v) = l^v$,  $l(s) \le l(s_0)$ for all $s \in \mathcal{S}$,  and $l(s) \le l^v$ for all $s \in \mathcal{S}^\perp(s_0)$. Due to Lemma \ref{lemma:opengraph}(i) and  $\mathcal{S}^\perp(s_0) \cap \mathcal{S}^\perp(s^v) = \emptyset$, we get
\begin{align*}
2l(s_0)\cdot l^v + l(s_0)(\lceil 2\sqrt{n}\rceil - l^v-l(s_0))\geq \sum_{s \in \mathcal{S}^\perp(s_0)}l(s)+ \sum_{s \in \mathcal{S}^\perp(s^v)}l(s)+ \sum_{s \in \mathcal{S}^a}l(s)   =  \sum_{s \in \mathcal{S}}l(s)=2n\,,
\end{align*}
and thus
\begin{align*}
l^v \geq \frac{2n}{l(s_0)} +l(s_0) -\lceil 2\sqrt{n}\rceil.
\end{align*}
This together with   Lemma \ref{lemma:opengraph}(ii),  \eqref{ineq:lowerboundnum},  and $\lceil t\rceil <t+1$ implies
\begin{align*}
F_{\mathrm{bond}}( G_{\mathrm{o}} ) &\geq 2(l^v+l(s_0)+  \#\mathcal{S}^\mathrm{a} \EEE ) \geq 2 \left(\frac{2n}{l(s_0)} +2l(s_0) -\lceil 2\sqrt{n}\rceil+ \#\mathcal{S}^\mathrm{a} \EEE\right) \geq 2\left(4\sqrt{n}-\lceil 2\sqrt{n}\rceil+  \#\mathcal{S}^\mathrm{a} \EEE \right)\\&>2\left(2\lceil2\sqrt{n}\rceil-\lceil 2\sqrt{n}\rceil  + \# \EEE \mathcal{S}^\mathrm{a} -2  \right) =2\lceil 2\sqrt{n}\rceil + 2(\#\mathcal{S}^\mathrm{a} -2)\,.
\end{align*}
Again, since $ 2(l^v+l(s_0)+\mathcal{S}^\mathrm{a} ) \EEE \in 2\mathbb{N}$ and,    $\#\mathcal{S}^\mathrm{a} \ge 1$ by    \eqref{ineq:lowerboundnum}--\eqref{ineq:cardestimateb},  the claim \eqref{eq: main claim} follows also in case (b). This finishes the proof.

Finally, we make the following observation: the argument along with Lemma \ref{lemma:construction}(ii) also shows that $\#\mathcal{S}^\mathrm{a} \ge 2$ would induce that   $G$ is not a ground state. From this and  \eqref{ineq:cardestimateb} we deduce  
\begin{align}\label{ineq:cardestimateb2}
\#\mathcal{S}  =  l^v + l(s_0)+1
\end{align}
 for the number of strata of a ground state $G$ with $\mathrm{span}(s_0) \subsetneq  V$. 
\end{proof}

Estimate \eqref{eq: slice} is related to proving an isoperimetric inequality with respect to the $l^1$-perimeter via slicing. We present a corresponding argument in the continuum setting in Appendix~\ref{appendix2}.

\subsection{The $n^{3/4}$-law}\label{sec:fluctuation}

We close   with a fluctuation estimate for minimizers.

\begin{proof}[Proof of Theorem \ref{th: 34law}]
Clearly, it is enough to prove the statement for $n \ge n_0$ for some $n_0 \in \N$. We use the notation of the previous proof. In particular, we choose $s_0$ and $l^v$ as done  before \eqref{ineq:lowerboundnum}.  As each $x \in V$  belongs to exactly two strata, by \eqref{ineq:cardstrata} we find that $l(s_0) = \max_{s \in \mathcal{S}} l(s) \le  \lceil 2\sqrt{n}\rceil$.
We start by noting that 
\begin{align}\label{eq: a bit different} 
l(s_0) \cdot l^v \geq n  - \lceil 2\sqrt{n}\rceil.
\end{align}
Indeed, if $\mathrm{span}(s_0) = V$, we have $l(s_0) \cdot l^v \geq n$. Otherwise, in view of \eqref{ineq:cardestimateb2}, the span missed \BBB exactly \EEE one stratum, consisting of at most $\lceil 2\sqrt{n}\rceil$ points. 

Now, by \eqref{eq: lengths_0}, \eqref{ineq:lowerboundnum},  and  \eqref{eq: a bit different}  we compute for $n$ sufficiently large
\begin{align*}
4 \sqrt{n} + 2 \ge 2\lceil 2\sqrt{n} \rceil   & = F(G) = 2 \# \mathcal{S} \ge 2 \big( l(s_0) + l^v  \big) \ge 2   l(s_0) + 2\frac{n  - \lceil 2\sqrt{n}\rceil}{l(s_0)}  \\ 
& \ge  2   l(s_0) + 2\frac{n}{l(s_0)}  -  \frac{(2\sqrt{n} + 1)^2}{n}  \ge    2   l(s_0) + 2\frac{n}{l(s_0)}  - 5\,.
\end{align*}
This yields
$$2   l(s_0)^2  - (4 \sqrt{n} + 7)l(s_0)  + 2n \le 0\,.   $$ 
Therefore
$$ x_- \le  l(s_0) \le x_+\,, $$
where
$$x_\pm = \frac{4 \sqrt{n} + 7 \pm \sqrt{(4 \sqrt{n} + 7)^2-16n}}{4}\,. $$
Then, a short computation yields
\begin{align}\label{341}\sqrt{n} - cn^{1/4} \le  l(s_0) \le \sqrt{n} + cn^{1/4}
\end{align}
for some universal $c >0$ large enough. Using again \eqref{eq: a bit different} and $l^v \le l(s_0)$, we also find
\begin{align}\label{342}
\sqrt{n} - cn^{1/4} \le l^v \le \sqrt{n} + cn^{1/4}
\end{align}
for a  larger $c>0$. In view of \eqref{ineq:cardestimateb2}, we get that, up to a translation and up to one stratum, $C_n$   is contained in the rectangular subset of $\Z^2$ defined by 
$$R_n:= \big\{(k_1,k_2) \colon \, k_1 \in \lbrace 1,\ldots, l(s_0)\rbrace, \ \ k_2 \in \lbrace 1, \ldots l^v \rbrace \big\} \,.$$
As each stratum consists of at most $\lceil 2\sqrt{n}\rceil$ points, we 
get 
\begin{align}\label{eq: LLL}
\# (C_n \setminus R_n) \le \lceil 2\sqrt{n}\rceil\,.  
\end{align}
Note that  \eqref{341}--\eqref{342} imply
$$\#  (R_n   \triangle S_n) \le cn^{3/4}. $$
Thus, recalling \eqref{eq: LLL}, to conclude it now suffices to prove that
\begin{align*}
 l(s_0) \cdot l^v   - n \le  5 \EEE \sqrt{n}\,.
 \end{align*}
 Assume by contradiction that  $l(s_0) \cdot l^v   - n > 5 \sqrt{n}$. 
Then, since $l(s_0) \le \lceil 2\sqrt{n}\rceil$, we get
\begin{align*}
4 \sqrt{n} + 2 &\ge 2\lceil 2\sqrt{n} \rceil    = 2 \# \mathcal{S} \ge 2 \big( l(s_0) + l^v  \big) > 2   l(s_0) + 2\frac{n}{l(s_0)} + 2\frac{ 5 \sqrt{n}}{l(s_0)} \ge 2   l(s_0) + 2\frac{n}{l(s_0)} + \frac{ 5}{2}
\end{align*}
for $n$ large enough. Since $2   l(s_0) + 2\frac{n}{l(s_0)}  \ge 4\sqrt{n}$, we obtain a contradiction. This concludes the proof. 
\end{proof}

\section*{Acknowledgements} 
This work was supported by the DFG project FR 4083/3-1 and by the Deutsche Forschungsgemeinschaft (DFG, German Research Foundation) under Germany's Excellence Strategy EXC 2044 -390685587, Mathematics M\"unster: Dynamics--Geometry--Structure. \BBB The research of LK was supported by the DFG through the Emmy Noether Programme (project number 509436910). \EEE

\bigskip

\section*{Conflict of interest}

The authors have no competing interests to declare that are relevant to the content of this article.

 Data sharing not applicable to this article as no datasets were generated or analysed during the current study.

\appendix 
\section{Proof of Lemma \ref{lemma:elementaryprop}}\label{appendix}

\begin{proof}[Proof of Lemma \ref{lemma:elementaryprop}] Let $C_n$ be a minimizer of \eqref{eq: main energy}. For simplicity, we write $G=(V,E)$ instead of $G_{\mathrm{nat}}=(V,E_{\mathrm{nat}})$ for the associated natural bond graph.   \\
\noindent \begin{step}{1} In this step, we show
\begin{align}\label{ineq:1epsregular}
|x-y| \ge  1-\varepsilon \text{ for all }  \lbrace x,y \rbrace \in E\,.
\end{align}
Define  
\begin{align}\label{def:M}
M:= \max_{x \in \mathbb{R}^2} \#(V \cap B_{\frac{1}{2}(1-\varepsilon)}(x))\,.
\end{align}
It suffices to  show $M=1$. Let $ x_0 \in \R^2$ be a maximizer. After translation of $V$, it is not restrictive to  assume  that $x_0=0$. By assumption {\rm ($\rm{iii}_2$)}  we have 
\begin{align}\label{estimate:insideBall}
\underset{\{x,y\} \in E}{\sum_{x,y \in B_{\frac{1}{2}(1-\varepsilon)}}} v_2(|x-y|)\geq \frac{1}{2\varepsilon}M(M-1)\,.
\end{align}
Consider the annulus $ A_\eps \EEE:=B_{\frac{1}{2}(1-\varepsilon) +\sqrt{2}} \setminus B_{\frac{1}{2}(1-\varepsilon)} \subset B_{\frac{1}{2}+\sqrt{2}}$. There exists $ N \in  \mathbb{N}$ and $\{z_i\}_{i=1}^{N} \subset \mathbb{R}^2$ such that  for all $0 < \eps \le \frac{1}{2}$ 
\begin{align*}
A_\eps  \subset B_{\frac{1}{2}+\sqrt{2}}  \subset  \bigcup_{i=1}^{N} B_{\frac{1}{4}}(z_i)   \subset \bigcup_{i=1}^{N} B_{\frac{1}{2}(1-\varepsilon)}(z_i)\,.
\end{align*}
 Thus, recalling \eqref{def:M}, we have
\begin{align}\label{ineq:cardannulus}
\#(V\cap A_\eps) \leq \#\left(V\cap  \bigcup_{i=1}^N B_{\frac{1}{2}(1-\varepsilon)}(z_i)\right)\leq \sum_{i=1}^N \#\left(V\cap  B_{\frac{1}{2}(1-\varepsilon)}(z_i)\right)\leq N M\,.
\end{align}
By \eqref{ineq:cardannulus}, the definition of $M$, and {($\rm i_2$)}, {($\rm ii_2$)} we have
\begin{align}\label{ineq:annulusinteraction}
\underset{\{x,y\} \in E}{\sum_{x \in B_{\frac{1}{2}(1-\varepsilon)}, y \in A_\eps}} v_2(|x-y|) \geq -1 \cdot \#\{(x,y) \colon x \in V\cap B_{\frac{1}{2}(1-\varepsilon)}, y \in V\cap A_\eps\} \geq  - NM^2\,.
\end{align}
We write $V\cap B_{\frac{1}{2}(1-\varepsilon)} = \{x_i\}_{i=1}^M$ and consider  a competitor $\hat{V}$ (with associated natural  bond graph $\hat{G}$) given by
\begin{align*}
\hat{V} = (V \setminus B_{\frac{1}{2}(1-\varepsilon)}) \cup \bigcup_{i=1}^M \{x_i +\tau_i\}\,,
\end{align*}
where $\tau_i \in \mathbb{R}^2$ are chosen such that
\begin{align}\label{ineq:choicetaui}
\mathrm{dist}(x_i+\tau_i, \hat{V}\setminus \{x_i+\tau_i\}) \geq \sqrt{2} \text{ for all } i=1,\ldots,M\,.
\end{align}
By   \eqref{ineq:choicetaui}, {($\rm ii_2$)},  and the optimality of $G$ we have
\begin{align}\label{ineq:Xmin}
F(G) \leq F(\hat{G}) &\leq F(G) - \underset{\{x,y\} \in E}{\sum_{x,y \in B_{\frac{1}{2}(1-\varepsilon)}}} v_2(|x-y|) - 2 \underset{\{x,y\} \in E}{\sum_{x \in B_{\frac{1}{2}(1-\varepsilon)}, y \in A_\eps}} v_2(|x-y|)\,.
\end{align}
 Now, using \eqref{ineq:Xmin}, \eqref{estimate:insideBall}, and \eqref{ineq:annulusinteraction}, we obtain
\begin{align*}
\frac{1}{2\varepsilon}M(M-1) \leq \underset{\{x,y\} \in E}{\sum_{x,y \in B_{\frac{1}{2}(1-\varepsilon)}}} v_2(|x-y|) \leq - 2  \underset{\{x,y\} \in E}{\sum_{x \in B_{\frac{1}{2}(1-\varepsilon)}, y \in A_\eps}} v_2(|x-y|) \leq 2  N M^2\,.
\end{align*}
For $\varepsilon>0$ small enough ($\varepsilon < \frac{1}{8N}  $ suffices), this inequality can only be true for $M=1$. This yields \eqref{ineq:1epsregular} and concludes Step 1.
\end{step} \\
\noindent \begin{step}{2} In this step we prove \BBB that \EEE all bond angles satisfy 
\begin{align}\label{incl:2epsregular}
\theta \in [\pi/2-\varepsilon,\pi/2+\varepsilon] \cup [\pi-\varepsilon,\pi+\varepsilon] \cup [3\pi/2-\varepsilon,3\pi/2+\varepsilon]\,.
\end{align}
In particular, for $\varepsilon< \frac{1}{10}\pi$, we then also have $\#\mathcal{N}(x,E) \leq 4$ for all $x \in V$ since  all bond angles at  $x \in V$ sum up to $2\pi$. To see \eqref{incl:2epsregular}, we   first of all  claim that 
\begin{align}\label{ineq:neighbourhood}
\#\mathcal{N}(x,E) \leq 4\frac{\left(\sqrt{2} +\frac{1}{2}\right)^2}{(1-\varepsilon)^2} \text{ for all } x\in V\,.
\end{align}
This follows by Step~1 and, due to {\rm ($\rm{ii}_2$)}, by the fact that $B_{\frac{1}{2}(1-\varepsilon)}(y) \subset B_{\sqrt{2}+\frac{1}{2}}(x)$ for all $y \in \mathcal{N}(x,E)$. More precisely,
\begin{align*}
\big(\sqrt{2} +\tfrac{1}{2}\big)^2\pi = | B_{\sqrt{2}+\frac{1}{2}}(x)| \geq \sum_{y \in \mathcal{N}(x,E)} |B_{\frac{1}{2}(1-\varepsilon)}(y)| \geq  \frac{1}{4}(1-\varepsilon)^2 \pi\,\# \mathcal{N}(x,E)\,,
\end{align*}
i.e., \eqref{ineq:neighbourhood} holds. Now, \eqref{incl:2epsregular} follows. In fact, if $x$ has a bond angle that does not satisfy \eqref{incl:2epsregular} we could define $\hat{V}= (V\setminus \{x\} )\cup \{x+\tau\}$ for some $\tau \in \mathbb{R}^2$ such that $\mathrm{dist}(x+\tau, \hat{V}\setminus \{x+\tau\})\geq \sqrt{2}$. Then, by  {($\rm i_2$)}, {($\rm iv_3$)}, and \eqref{ineq:neighbourhood}  we obtain a contradiction to the minimality of $G$. Summarizing, with the choice $\varepsilon_0 : = \min\{\frac{1}{10}\pi, \frac{1}{8N}\}$,  the statement holds. 
\end{step}
 \end{proof}

\section{Proof of isoperimetric inequalities in $l^1$ via slicing}\label{appendix2}

In this short excursion, we show how isoperimetric inequalities with respect to the $l^1$-perimeter can be obtained by a slicing argument similar to the one used in case (a) of the proof of Theorem~\ref{thm:crystallization}. Indeed, our main proof was inspired by such an argument.  We present it directly in any space dimension $d \ge 1$.  First, given $m \in \mathbb{N}$ and $x\in \mathbb{R}^m$ we denote by $| x|_1= \sum_{k=1}^m |x_k|$ its $l^1$-norm.  Now, for a set of finite perimeter $E \subset \R^d$, see \cite{maggi2012sets}, we introduce the \emph{$l^1$-perimeter} by 
$$P^d_{l^1}(E) = \int_{\partial^* E}   | \nu_E|_1 \, {\rm d}\mathcal{H}^{d-1},  $$
\BBB where $\partial^* E$ denotes the reduced boundary and $\nu_E$ denotes its measure theoretical outer normal. \EEE

\begin{theorem}\label{thm: slic}
For each set of finite perimeter $E \subset \R^d$, $d \ge 1$, it holds that
$$P^d_{l^1}(E) \ge 2d (\mathcal{L}^d(E))^{1-1/d}.$$
\end{theorem}

A proof of this result via slicing hinges on the following lemma, for which we use the notation  
\begin{equation}
E_t = E \cap \{ (x',t)\colon x' \in \mathbb{R}^{d-1} \} \quad \quad \text{for all $t  \in \R$}\,.
\end{equation}

\begin{lemma}\label{lem:symmetrization}
Suppose that $E \subset \R^d$ is a bounded set of finite perimeter. Then,
\begin{equation}
P^d_{l^1}(E) \geq \int_{\mathbb{R}} P^{d-1}_{l^1}(E_t) \, {\rm d}t + 2 \sup_{t \in \R} \mathcal{H}^{d-1}(E_t).
\end{equation}
\end{lemma}

We postpone the proof of the lemma to the end.

\begin{proof}[Proof of Theorem \ref{thm: slic}]
We prove the statement by induction. The case $d=1$ is clear as each set $E \subset \R$ with finite volume satisfies $P^1_{l^1}(E) \ge 2$. Suppose that the statement holds for $d-1$ and consider $E \subset \R^d$. Then, by Lemma \ref{lem:symmetrization}  and the induction hypothesis we have 
\begin{align*}
P^d_{l^1}(E) &\geq \int_{\mathbb{R}} P^{d-1}_{l^1}(E_t) \, {\rm d}t + 2 \sup_{t \in \R} \mathcal{H}^{d-1}(E_t) \\
& \geq \int_{\mathbb{R}}  2(d-1) \big(\mathcal{H}^{d-1}(E_t)\big)^{1-1/(d-1)} \, {\rm d}t + 2 \sup_{t \in \R} \mathcal{H}^{d-1}(E_t) \,.
\end{align*}
Using the shorthand $M:=  \sup_{t \in \R} \mathcal{H}^{d-1}(E_t)$ and integrating over the slices $E_t$ we get
\begin{align*}
P^d_{l^1}(E) \ge \int_{\mathbb{R}}  2(d-1)  M^{-1/(d-1)} \mathcal{H}^{d-1}(E_t)   \, {\rm d}t + 2 M = 2(d-1)  M^{-1/(d-1)} \mathcal{L}^d(E)  + 2 M\,.
\end{align*}
By optimizing with   respect to $M$ we get $M =  (\mathcal{L}^d(E))^{1 - 1/d}$, and  thus we conclude    
$$P^d_{l^1}(E) \ge 2d (\mathcal{L}^d(E))^{1-1/d}\,. $$
\end{proof}

\begin{proof}[Proof of Lemma \ref{lem:symmetrization}]
We start by  splitting the $l^1$-perimeter into 
\begin{equation}\label{eq:estimateonperimeter}
P_{l^1}(E) = \int_{\partial^*E} |\nu_E|_1 \, {\rm d}\mathcal{H}^{d-1} = \int_{\partial^*E} (|\nu'_E|_1 + |(\nu_E)_d|) \, {\rm d}\mathcal{H}^{d-1},
\end{equation}
where $\nu'_E = ((\nu_E)_1, \ldots, (\nu_E)_{d-1})\in \R^{d-1}$. Introducing the function
\begin{equation*}
\bar{g} = \frac{|\nu'_E|_1}{\sqrt{1 - |(\nu_E)_d|^2}} \quad \text{ on $\partial^* E$}\,,
\end{equation*}
the coarea formula, see \cite[(18.25)]{maggi2012sets},  implies 
\begin{equation}\label{eq:estimateofg1}
\int_{\partial^* E} |\nu'_E|_1 \, {\rm d}\mathcal{H}^{d-1} = \int_{\partial^* E} \bar{g} \, \sqrt{1 - (\nu_E)_d^2} \, {\rm d}\mathcal{H}^{d-1} = \int_{\mathbb{R}} \int_{(\partial^* E)_t} \bar{g} \, {\rm d}\mathcal{H}^{d-2}  {\rm d} t = \int_\mathbb{R} P^{d-1}_{l^1}(E_t) \, {\rm d}t,
\end{equation}
where in the last step we used the fact that 
\BBB $(1 - (\nu_E)_d^2)^{-1/2} \nu_E' \in \R^{d-1}$ \EEE is a unit normal to $E_t$. On the other hand, using the notation $ (\partial^* {E})^{x'} :=   \partial^*E  \cap \{ (x',t): t \in \mathbb{R}\}$ for $x' \in \R^{d-1}$, by slicing properties of $BV$-functions, we obtain
\begin{equation}\label{eq:estimateofg2}
\int_{\partial^* E} |(\nu_E)_d| \, {\rm d}\mathcal{H}^{d-1} = \int_{\R^{d-1}}  \mathcal{H}^0\big( (\partial^* {E})^{x'} \big) \, {\rm d}\mathcal{H}^{d-1}(x') \geq 2 \sup_{t \in \R} \mathcal{H}^{d-1}(E_t),
\end{equation}
where we used that for  $\mathcal{H}^1$-a.e.\ $t \in \R$ and  $\mathcal{H}^{d-1}$-a.e.\  $x'$ with  $(x',t) \in E_t$ we have  $ \mathcal{H}^0( (\partial^* {E})^{x'} ) \ge 2$. Combining the two estimates \eqref{eq:estimateofg1}--\eqref{eq:estimateofg2}, the desired results follows from  \eqref{eq:estimateonperimeter}.
\end{proof}


\end{document}